\documentclass[12pt]{iopart}

\usepackage{amsfonts}
\usepackage{amssymb}
\usepackage{amsthm}

\usepackage{e uscript}
\usepackage{tikz}
\usetikzlibrary{calc}
\usetikzlibrary{arrows}

\theoremstyle{plain}
\newtheorem{theorem}{Theorem}[section]
\newtheorem{lemma}[theorem]{Lemma}
\newtheorem{proposition}[theorem]{Proposition}
\newtheorem{corollary}[theorem]{Corollary}

\theoremstyle{definition}

\theoremstyle{remark}
\newtheorem{remark}[theorem]{Remark}

\begin{document}

\title[Watermelons on the half-plane]{Watermelons on the half-plane}
\author{Kh.D. Nurligareev$^{1,3}$,
 A.M. Povolotsky$^{2,3}$}
\address{$^1$ LIPN, University Sorbonne Paris Nord, F-93430, Villetaneuse, France}
\address{$^2$ Bogoliubov Laboratory of Theoretical Physics, Joint Institute for
Nuclear Research, 141980, Dubna, Russia}
\address{$^3$ National Research University Higher School of Economics, 101000, Moscow, Russia}
\eads{\mailto{$^1$khaydar.nurligareev@lipn.univ-paris13.fr}, \mailto{$^2$alexander.povolotsky@gmail.com}}

\begin{abstract}
 We study the watermelon probabilities in the uniform spanning forests on the two-dimensional semi-infinite square lattice near either open or closed boundary to which the forests can or cannot be rooted, respectively.
 We derive universal power laws describing the asymptotic decay of these probabilities with the distance between the reference points growing to infinity, as well as their non-universal constant prefactors.
 The obtained exponents match with the previous predictions made for the related dense polymer models using the Coulomb Gas technique and Conformal Field Theory, as well as with the lattice calculations made by other authors in different settings.
 We also discuss the logarithmic corrections some authors argued to appear in the watermelon correlation functions on the infinite lattice.
 We show that the full account for diverging terms of the lattice Green function, which  
ensures the correct probability normalization,  provides the pure power law decay in the case of semi-infinite lattice with closed boundary studied here, as well as in the case of  infinite lattice discussed elsewhere.
 The solution is based on the all-minors generalization of the Kirchhoff matrix tree theorem, the image method and the developed asymptotic expansion of the Kirchhoff determinants.  
\end{abstract}
\maketitle

\section{Introduction}\label{sec: Intro}

An interest to the problem of  spanning trees (ST) on  graphs goes back to the renowned paper by  G. Kirchhoff,  where he proved  what is now known as the Matrix Tree Theorem~\cite{Kirchhoff1847}. Originated from the theory of electric circuits, this theorem   gives a way of counting ST of  a given graph by manipulating  with the matrix of discrete Laplacian constructed out of the graph. Since  then, the subject of ST developed significantly, having become a substantial part of the graph theory.

Being initially a combinatorial object, ST are naturally incorporated into the framework of statistical physics and  probability theory, where they are considered as randomly chosen from the set of ST of a given graph according to a prescribed probability distribution. An example is the model of uniform spanning tree (UST) assigning the same probability to every spanning tree of a given graph or, more generally, the weighted uniform spanning tree model (WUST), in which every particular spanning tree is assigned a probability proportional to the product of the weights of the edges constituting the tree.

In the following, we mention only a few of many known applications. The bijection between perfect matchings or dimer packings on certain graphs with ST on their subgraphs~\cite{Temperley1973,Priezzhev1985} allows one to  study the Gibbsian measures on dimer configurations  in terms of the statistics of UST or, more generally, WUST models. Also, the UST model can be reformulated as a particular limit of the Fortuin-Kasteleyn random cluster model~\cite{FortuinKasteleyn1972}. By further mapping of the latter to the $q$-component Potts model, the random ST can be thought of as a formal $q\to 0$ limit of the Potts model~\cite{Stephen1976}. Another important  application to the non-equilibrium statistical physics appeared in the theory of self-organized criticality~\cite{Bak2013}, where a bijection between the set of ST and the recurrent set of configurations of the paradigmatic model of the theory, Abelian sandpile model (ASM), was established~\cite{MajumdarDhar1992}. Thus, the stationary measure of ASM can be identified with the UST measure.

It is also worth mentioning important connections with the theory of Markov processes. A Markov chain with an arbitrary   transition graph can be recast as the random walk on the set of ST of this graph. Thus, the stationary measure of the chain induces a measure on this set. This fact was in the core of the early algorithms for generating random ST with a prescribed distribution~\cite{Broder1989, Aldous1990,ProppWilson1998}. In particular, simple random walks (RW) can be used for sampling the UST. However, the most efficient algorithm proposed by Wilson~\cite{Wilson1996} exploited the relation of the UST model with yet another Markov process, the loop erased random walk (LERW).  Conversely, many results obtained on the UST  can be translated into the statistics of the LERW. In particular, the distribution of a path on a spanning tree is the same as the one of the~LERW~\cite{Pemantle1991}.

In the framework of probability theory, the consideration of ST can be naturally extended to infinite graphs. It was shown~\cite{BenjaminiLyonsPeresSchramm2001} that for many graphs there exists an infinite graph limit of this measure called the uniform spanning forests (USF). A~reasonable question in this context is how a typical spanning tree (or forest) looks like. What is the statistics of the local as well as the large-scale events within this random object? Over the last twenty years, these questions were the subject of extensive studies, which culminated in  many bright results.

The primary interest of the physical community to  the UST-related models on infinite lattices was due to the fact that they gave relatively  simple examples of exactly solvable systems at criticality that revealed itself in the long range correlations. Specifically, the  critical exponents and scaling functions characterizing the large-scale universal behavior of the models on regular $d$-dimensional lattices could be obtained from combinatorial manipulations with Kirchhoff theorem~\cite{Priezzhev1994}. On the other hand, since the mid-eighties of the twentieth century, there was an understanding that the continuous limit of UST in two dimensions was  described by the $c=-2$ conformal field theory (CFT)~\cite{Duplantier1986, Ivashkevich1999}. This correspondence was later put on the firm mathematical ground in the framework of the Schramm-Loewner evolution (SLE) approach~\cite{LawlerSchrammWerner2011}. Therefore, the lattice calculations serve as a verification of CFT predictions and vice versa. A~remarkable fact is that the correlations observed in the UST-related models fall not only in the realm of usual $c=-2$ CFT but also in its logarithmic version~\cite{Cardy2013}. Still, there are no many examples of exact lattice calculations available to illustrate the latter.

One of the examples is the two-point height distribution in the ASM. The joint probability of minimal heights at two points in the bulk of the lattice separated by a~large distance shows asymptotically a power law distance dependence~\cite{MajumdarDhar1991}. So do  the two point height probabilities associated with points near the closed or open boundary of the lattice for any values of the heights \cite{BrankovIvashkevichPriezzhev1993, Ivashkevich1994}. In contrast, when one of heights is greater than minimal, the power law asymptotics acquires the logarithmic prefactor \cite{PoghosyanGrigorevPriezzhevRuelle2008, PoghosyanGrigorevPriezzhevRuelle2010}. It was interpreted as a lattice analogue of the CFT correlation functions of certain fields and their logarithmic partners~\cite{Ruelle2013}.  

In the language of ST on the lattice, the logarithmic corrections to the power law asymptotics come from counting certain non-local events. Specifically, the derivation of height probabilities by Priezzhev is based on counting  non-local spanning tree sub-configurations coined the theta-graphs \cite{Priezzhev1994}. It is the interaction  of the theta-graph with distant defects inserted into the lattice that is responsible for the logarithm in the correlation functions.

The theta-graph itself is the smallest $k=3$ example of the $k$-leg watermelon, that is, $k$ disjoint paths on the spanning tree starting and terminating in two distant groups of closely spaced vertices. The question of whether a logarithm is present in the watermelon probability itself has been debated for some time. 

In the late eighties of the twentieth century, the $\Or(n)$ loop model, as well as related $\Or(n)$ vector model, Potts model, polymer models, percolation models, self-avoiding walks and some other models were extensively studied~\cite{Duplantier1989} with the methods of Coulomb Gas (CG) theory~\cite{Nienhuis1987} and CFT~\cite{Cardy1987}. In particular, a collection of critical exponents describing power law decay of the watermelon correlation functions low-temperature phase of the $n\to 0$ limit of $\Or(n)$ model, as well as tightly related with it ST and dense polymer models, were predicted, both in the bulk of the infinite plane and near the boundary of semi-infinite half-plane~\cite{Duplantier1986, DuplantierSaleur1986, DuplantierSaleur1987, Duplantier1987, DuplantierDavid1988}. All these models belong to the same universality class and, hence, have the same power law large-scale behavior.  

Watermelons in UST on bounded regions  of the square lattice were studied by Kenyon~\cite{Kenyon2000}. He obtained the asymptotics of so-called crossing probability on the rectangular domain of the square lattice with either free or periodic boundary conditions. In these almost one-dimensional geometries, the correlations decay exponentially as the lattice length grows to infinity. Kenyon also proved the conformal invariance of the results, which allows one to associate the results on the strip and cylinder with those for the semi-annulus and annulus on the plane, respectively, where the exponential decay turns into the power laws with exponents matching with the Coulomb Gas predictions.

Another attempt to study the watermelon correlation functions right on the infinite square lattice was undertaken by Ivashkevich and Hu~\cite{IvashkevichHu2005}. In addition to the power law with the exponent predicted from CG theory for the infinite plane and found by Kenyon for the annulus, their asymptotic formula also possessed a logarithmic prefactor. Gorsky, Nechaev, Poghosyan and Priezzhev~\cite{GorskyNechaevPoghosyanPriezzhev2013}, who elaborated the arguments of~\cite{IvashkevichHu2005}, confirmed this result. 

It is of interest to further clarify the behavior of watermelon-related correlation functions, as well as to consider the same problem in different settings. In the present paper, we bring the setting of \cite{IvashkevichHu2005, GorskyNechaevPoghosyanPriezzhev2013} to the half-infinite square lattice. We consider the USF on the half-lattices $\mathbb{Z}\times\mathbb{Z}_{\geqslant0}$ and $\mathbb{Z}\times\mathbb{Z}_{>0}$ with open and closed boundary conditions (BC) at the lowest row, respectively, which imply that the forest components can or cannot be rooted to its sites. In addition, we insert a string of $k$ auxiliary roots near this row and a similar group of $k$ sites at a horizontal distance $r$ from the roots. We are interested in the probability for $k$ components corresponding to the roots to be connected to the given sites, i.e. to form a watermelon of the length $r$. As a result, we obtain two power laws for open and closed boundaries with $k$-dependent exponents and also the non-universal constant coefficients of the leading asymptotics. 


Note that the system near the closed boundary   is in a sense similar to the situation in the bulk considered in~\cite{IvashkevichHu2005, GorskyNechaevPoghosyanPriezzhev2013}, where the logarithmic factor was claimed to be a~signature of the logarithmic CFT. In the course of our derivation, we also reconsider those results. As we discuss in Remark~\ref{rem}  below, the quantity evaluated in~\cite{IvashkevichHu2005, GorskyNechaevPoghosyanPriezzhev2013} was actually not the probability, but a finite part of the infinite ratio of the number of watermelons to the number of ST. At the same time, a suitably normalized quantity defined to have a meaning of the watermelon probability has a pure power law asymptotics both in our case and in the bulk situation of ~\cite{IvashkevichHu2005, GorskyNechaevPoghosyanPriezzhev2013}.

Also, as we have mentioned above, many  statements about the UST model can be interpreted in the language of the LERW.  Likewise, we interpret the results obtained from the analysis of USF as a specially conditioned probability for $k$ LERW to connect specified sites near open or closed boundary of the half-infinite lattices. Similar constructions were studied by Fomin in~\cite{Fomin2001}, where the total positivity of so-called walk and hitting matrices was proved. The explanation of the total positivity was the fact that their minors were conditioned LERW partition functions. The conjecture made in~\cite{Fomin2001} that the statements survive the scaling limit and also hold for the Brownian motion was proved in~\cite{KozdronLawler2005}. In particular, the crossing exponents proved earlier by Kenyon for UST and LERW were reproduced. The ideas used in the above papers, were later developed in works of Kenyon and Wilson~\cite{KenyonWilson2011} and Karrila, Kyt\"ol\"a and Peltola~\cite{KarrilaKytolaPeltola2020}. They considered the probabilities of general boundary visit events for planar LERW, as well as connectivity events for branches in UST, and proved their convergence to the formulas of SLE theory. Those results have much larger generality than ours, and the exponents we obtain, in principle, should follow from their general results.

The rest of the article is organized as follows. In Section~\ref{sec: LERW and Watermelons}, we introduce necessary definitions, state main theorems and discuss their relation with earlier results. In Section~\ref{sec: Method of solution in general}, we discuss tools used below that are based on the generalization of the Kirchhoff matrix tree theorem and give necessary details about the Green functions for the infinite lattice and half-lattices that we consider. As a result, the  probability we are looking for is represented in the form of a determinant of a matrix of special structure. The asymptotic evaluation of such determinants is  performed in Section~\ref{sec: Counting the determinants}. The obtained formulas are applied to  particular cases under consideration in Section~\ref{sec: Isotropic Watermelons} that completes the proof of the main statements.

\section{Spanning forests, LERW and Watermelons}\label{sec: LERW and Watermelons}

\subsection{Definitions and results}

Let $\mathcal{G}=(V,E)$ be a finite undirected connected graph without self-loops and multiple edges. We distinguish a subset  of sites, $\partial\subset V,$ called (open or Dirichlet) \emph{boundary}. A~\emph{spanning forest} (SF) on $\mathcal{G}$ rooted to the boundary is a subgraph  of $\mathcal{G}$ containing all the vertices and no edge cycles, such that every its connected component includes exactly one site of the boundary (then we say that the component is \emph{rooted} to that site).
 \footnote{The adjectives ``open'' or ``Dirichlet'' applied to the boundary  emphasize the fact that SF can be rooted to its sites. Also, in context of the semi-infinite lattices, it marks the spacial boundary of the half-plane. In contrast,   the terms ``closed'' or ``Neumann'' boundary used below do not imply a boundary in the SF sense, rather having only the spacial meaning. Which one is meant in each case will be clear from the context and should cause no ambiguity. }
 Some connected components, which we refer to as \emph{empty}, may consist of a~single  boundary site. The edges of every non-empty component can be given a natural orientation towards the boundary. Therefore, we will often say about the directed forest. The SF rooted to the boundary consisting of a single site are ST.
 \footnote{It is customary in the literature to consider the wired BC, which suggest that the boundary consists of a single site. 
This is an equivalent formulation obtained by gluing  all the boundary sites to one, so that  only ST can exist on a finite graph.  
Below we keep to the multi-site boundary for further notational convenience.}

Denote the set of SF rooted to $\partial$ by $\mathfrak{F}(\partial)$. One can define a probability measure on~$\mathfrak{F}(\partial)$ by assigning a weight $w(e)$ to every edge $e\in E$. In this case, the probability of a~spanning forest $\mathcal{F}$ is proportional to the product of weights of its edges,
\begin{equation}\label{weight}
 W(\mathcal{F}) = \prod_{e\in\mathcal{F}}w(e).
\end{equation}
In other words, this probability is equal to
\begin{equation}\label{eq: sfweight}
 \mathbb{P}_\mathcal{G}^{\partial}(\mathcal{F}) = \frac{W(\mathcal{F})}{Z_\mathcal{G}(\partial)},
\end{equation}
where 
\begin{equation}\label{eq: Zf}
Z_\mathcal{G}(\partial)=\sum_{\mathcal{F}\in\mathfrak{F}(\partial)} W(\mathcal{F}),
\end{equation}
is the normalization factor referred to as  partition function. In particular, when all the weights are equal, say $w(e)=1$ for any $e\in E$, we obtain the model of uniform rooted~SF. In general, we say about the weighted rooted SF.

In the following, we consider SF that have additional components rooted to some fixed set of sites. In particular, we want to control the components that connect specific sites to their roots. To this end, let us  introduce notations containing this extra information. Given $k\geqslant 0$ and  $n\geqslant 0$, let $I=(i_1,\dots,i_k)$, $J=(j_1,\dots,j_k)$ and $R=(r_1,\dots,r_n)$ be three disjoint groups of sites. Then we denote
\[
 \mathfrak{F}(IJ|R|\partial) \equiv
 \mathfrak{F}(i_1j_1|\dots|i_kj_k|r_1|\dots|r_n|\partial)
\]
the set of forests such that for any $\mathcal{F}\in\mathfrak{F}(IJ|R|\partial)$:
\begin{itemize}
 \item
  every component of $\mathcal{F}$ is rooted to $I\cup R\cup\partial$,
 \item
  for every $l=1,\ldots,k$, the site $j_l$ belongs to the component rooted to $i_l$.
\end{itemize}
 The corresponding partition function is
\[
 Z_{\mathcal{G}}(IJ|R|\partial) =
 \sum_{\mathcal{F}\in\mathfrak{F}(IJ|R|\partial)} W(\mathcal{F}).
\]
Note that if $k=0$ or $n=0$, then the corresponding sets are empty and we omit them in the notation. Thus, we write $\mathfrak{F}(R|\partial)$ for $k=0$, as well as $\mathfrak{F}(IJ|\partial)$ for $n=0$.

We are interested in specific SF called watermelons. Given two distant sets $I$ and $J$ of $k$ closely spaced sites each, a $k$-\emph{watermelons} (embedded into the forest rooted to the boundary $\partial$) is an element from the set $\mathfrak{F}(IJ|\partial)$. The set of watermelons is a subset of the set of all forests rooted to the extended boundary $\partial \cup I$, that is, $\mathfrak{F}(IJ|\partial)\subset\mathfrak{F}(I,\partial)$. The measure on such forests is defined by (\ref{eq: sfweight}-\ref{eq: Zf}) with the boundary $\partial\cup I$ instead of~$\partial$. In particular, the watermelon probability, which is the main subject of our interest, is
\begin{equation}\label{eq: watermelonprob}
\mathbb{P}_{\mathcal{G}}^{I\cup\partial}(\mathfrak{F}(IJ|\partial))=\frac{Z_{\mathcal{G}}(IJ|\partial)}{Z_{\mathcal{G}}(I|\partial)}.
\end{equation}


As it has been discussed in the introduction, the uniform or weighted ST are tightly related to the LERW problem, so the probabilistic results obtained on ST answer questions related to LERW. To explain the connection between our problem and LERW, let us recall \emph{Wilson's algorithm} of generating weighted rooted ST~\cite{Wilson1996}. Given a finite graph $\mathcal{G}$ with edge weights being real positive numbers, the algorithm is as follows:
\begin{enumerate}
 \item
  Fix the root vertex $v_1$ and enumerate the other vertices of $\mathcal{G}$ the way you like: $V=\{v_1,\ldots,v_N\}$.\par
 \item
  Define $\EuScript{U}_1=\{v_1\}$ to be the set consisting of the single vertex $v_1$.\par
 \item
  For every integer $k=2,\ldots,N$, consider a weighted loop-erased random walk $LE[\pi(v_k,\EuScript{U}_{k-1})]$, which is the trajectory of a weighted random walk $\pi(v_k,\EuScript{U}_{k-1})$ that starts at $v_k$, stops having reached $\EuScript{U}_{k-1}$ and has all its loops erased in chronological order. A weighted random walk is a Markov chain on $V$ with transition probabilities $p(u,v)$ (from $u\in V$ to  $v\in V$) which are normalized weights,
 \[
  p(u,v) = \frac{w(u,v)}{\sum_{t\in V}w(u,t)}.
 \]
Given a trajectory $LE(\pi(v_k,\EuScript{U}_{k-1}))$, define $\EuScript{U}_k$ to be the union of $\EuScript{U}_{k-1}$ and the set of trajectory vertices.\par
\item Define the spanning tree $\mathcal{T}$ to be the union
 \[
  \mathcal{T} =
  \bigcup_{i=1}^{N-1}
   LE[\pi(v_i+1,\EuScript{U}_{i})].
 \]
\end{enumerate}

Thus, this algorithm generates a weighted    spanning tree on a graph rooted to a~single boundary site $v_1$. Given a graph $\mathcal{G}$ with a multisite boundary $\partial$, we can initially define $\EuScript{U}_1=\partial$ and then run the rest of the algorithm (step 3) unchanged. As a result, we obtain the weighted SF rooted to $\partial$.

What is a watermelon from the set $\mathfrak{F}(IJ|\partial)$ in terms of LERW?   Let us enumerate the sites in $V\backslash (\partial\cup I)$, so that the first $k$ sites in the list are $j_1,\dots,j_k.$ Then, we apply Wilson's algorithm on $\mathcal{G}$ with boundary $\partial\cup I$ running LERW from every site in the list subsequently. To obtain a watermelon, every of the first $k$ trajectories of the LERW, starting from $j_1,\dots,j_k$, should stop at $i_1,\dots,i_k$, respectively, so that each of the corresponding random walks does not touch all the previous LERW or the other boundary sites on the way. The further application of Wilson's algorithm constructs~SF preserving the first $k$ trajectories. Thus, relation \eref{eq: watermelonprob} gives us the probability of the sequence of $k$ non-intersecting LERW trajectories starting from $j_1,\dots,j_n$ to first reach the set $I\cup\partial$ at sites $i_1,\dots,i_k$, respectively, so that the random walks $\pi
(j_l,i_l)$ do not reach  the preceding LERWs  as well as the boundary, i.e. 
\[
 \pi(j_l,i_l) \bigcap
  \{LE[\pi(j_{1},i_{1})], \ldots, 
   LE[\pi(j_{l-1},i_{l-1})], \partial\} = 
 \emptyset,\quad  l=2,\dots,k.
\]

Now we give a precise meaning to the problem of watermelon probabilities under consideration. Let $\mathcal{L}^{\mathrm{op}}=(V^\mathrm{op},E^\mathrm{op})$ and $\mathcal{L}^{\mathrm{cl}}=(V^\mathrm{cl},E^\mathrm{cl})$ be  two instances of the semi-infinite square lattice, which will be referred to as the half-lattices with open and closed BC, respectively. Here $V^{\mathrm{op}} = \mathbb{Z}_{\geqslant 0}\times \mathbb{Z}$ and $V^\mathrm{cl} = \mathbb{Z}_{>0}\times \mathbb{Z}$, respectively, and
\[
 E^{\mathrm{op},\mathrm{cl}} = 
 \{e=(v,v+\mathbf{e}_i) \mid v\in V^{\mathrm{op},\mathrm{cl}},i=1,2\}.
\]
 In the former (open) case, the lowest raw consists of boundary sites $\partial^{\mathrm{op}}=\{(k,0)\}_{k{\in\mathbb{Z}}}$ in the sense defined above.

We fix two sets of roots and endpoints of watermelons to be
\begin{eqnarray}\label{eq: I,J}
 I=\{(i,1)\}_{1\leqslant i\leqslant k }, \quad
 J=\{(i+r,1)\}_{1\leqslant i\leqslant k }. 
\end{eqnarray}
 Note that available analytic tools are limited to very special choices of the sets  $I,J$. This issue was discussed in detail in the paper of Fomin \cite{Fomin2001} in the context of total positivity of matrices. In particular, Fomin considered pairings of sites of a graph with LERW, similar to our watermelons, and proved the determinantal formulas for signed sums of their weights   over different pairings extending the Karlin-McGregor~\cite{KarlinMcGregor1959} and Lindtr\"om-Gessel-Viennot~\cite{Lindstrom1973, GesselViennot1985} theorems to the case of LERW.

Similar formulas appear below in the context of SF from the  generalization of Kirchhoff theorem. For these sums to have a meaning of probabilities, all their summands should at least be non-negative. This corresponds to the positivity conditions studied in~\cite{Fomin2001}. In simple terms, geometric constraints should prevent pairings with a wrong sign. One of the possible constraints is used for the watermelon in the bulk of infinite square lattice in~\cite{IvashkevichHu2005, GorskyNechaevPoghosyanPriezzhev2013}, where  the sets $I$ and $J$ are taken to be \emph{zigzags} with odd number $k$ of points (see Figure~\ref{fig: fences}). The oddness of $k$  guarantees that the obtained sum counts the forests of two types, both with the same sign. Namely, in~\cite{IvashkevichHu2005, GorskyNechaevPoghosyanPriezzhev2013} the following types are counted: $\mathfrak{F}(i_1j_{1}|\ldots|i_kj_{k}|\partial)$ and $\mathfrak{F}(i_1j_{\sigma(1)}|\ldots|i_kj_{\sigma(k)}|\partial)$, where permutation $\sigma$ is a long cycle. In our case, the strings of sites are located near the boundary of the half-lattice, and hence, only the watermelons of the form $\mathfrak{F}(i_1j_{k}|\ldots|i_kj_{1}|\partial)$ are geometrically possible.

\begin{figure}[ht]
 \centerline{\includegraphics[width=0.35\textwidth]{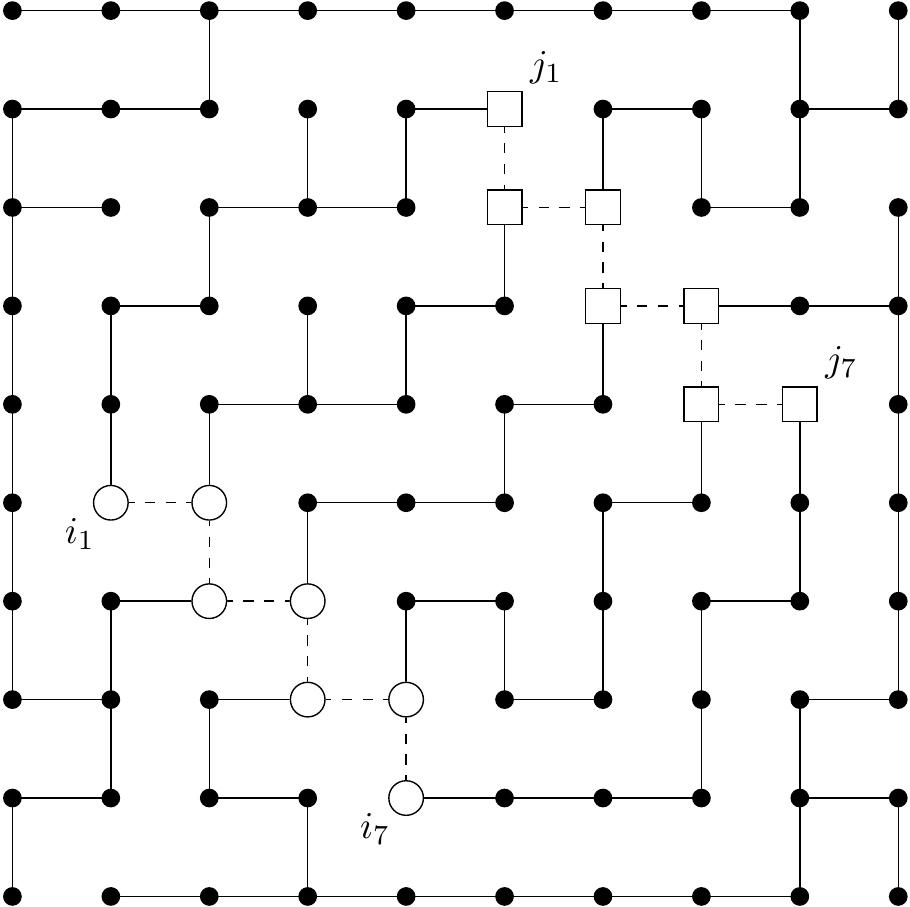}}
 \caption{Watermelon in the bulk of the square lattice with $k=7$ legs. The root set $J=\{i_1\dots,i_7\}$ ($\circ$) and  endpoint set $J=\{j_1,\dots,j_7\}$ ($\openbox$) are zigzag-like fences shown by circles connected with the dashed lines. Thick black lines show the corresponding LERWs.} \label{fig: fences}
\end{figure}

Another thing, which is yet to be specified when the infinite lattice is considered, is  how the infinite lattice limit is reached. Below we consider  probabilities of events within SF on an infinite lattice $\mathcal{L}$ obtained as  limits of  probabilities of events associated with  forests on its finite  subsets from an exhausting sequence $\{\mathcal{L}_n\}_{n\in \mathbb{N}}$ of finite subsets of $\mathcal{L}$. More precisely, let $\mathcal{L}_1\subset\mathcal{L}_2\subset\dots\subset\mathcal{L}$ and $\bigcup_{i\in\mathbb{N}}\mathcal{L}_i=\mathcal{L}$, where each subset $\mathcal{L}_n$ has its own boundary $\partial_n$, possibly going away to infinity in the limit $n\to\infty$. Then we say that the limiting measure (in the sense of weak convergence) exists if the probabilities  of local events  converge and that it is unique if it does not depend on the way the limit is taken, i.e. neither on the sequence $\{\mathcal{L}_n\}$ nor on BC on the  subsets of this sequence. The example is the USF measure on the infinite square lattice, as well as other graphs, where the random walk is recurrent. The existence and uniqueness of this measure are proved in~\cite{BenjaminiLyonsPeresSchramm2001}.

Though the existence of a watermelon is not a local event, it can be formalized as a countable union of local events, each having a probability assigned via the same     limiting procedure. Specifically, in our case, the probability of a local event, say denoted by $A$, is defined as a limit 
\begin{equation}\label{eq: limit prob}
 \mathbb{P}_{\mathcal{L}}^{I\cup\partial}(A)  = \lim_{n\to\infty} \mathbb{P}_{\mathcal L_n}^{I\cup\partial_n}(A)	
\end{equation}
taken over the exhausting sequences $\{\mathcal{L}_n\}_{n\geqslant1}$ of connected  lattice subsets of $\mathcal{L}$. In the case of closed BC, $\mathcal{L}=\mathcal{L}^\mathrm{cl}$ has the (extended) open boundary $I\bigcup \partial_n$, where the set $\partial_n=\partial_n^\mathrm{cl}$ consists of  the sites of $\mathcal{L}_n$ connected to sites of $\mathcal{L}$ outside of $\mathcal{L}_n$. In the case of open BC, $\mathcal{L}=\mathcal{L}^\mathrm{op}$, and we need to include the part of the lowest row $\partial^\mathrm{op}\bigcup\mathcal{L}_n$ to the boundary $\partial_n=\partial_n^\mathrm{op}$, see Figure~\ref{fig: boundary watermelons}. The symbol $\partial$ in the l.h.s. of (\ref{eq: limit prob}) is the limiting boundary, which  is either $\partial^\mathrm{op}$ or empty for $\mathcal{L}=\mathcal{L}^\mathrm{op}$ and $\mathcal{L}=\mathcal{L}^\mathrm{cl}$, respectively, and we omit the dependence on the boundary at infinity in the infinite lattice notations implying the described procedure. Such defined BC at infinity are often referred to as wired BC.  Note that the arguments of~\cite{BenjaminiLyonsPeresSchramm2001}  based on the recurrence are
  applicable to our case and   suggest that the  limit exists and does not depend on the choice of boundary conditions at infinity, i.e. is not limited to the wired BC choice for finite sub-graphs.

\begin{figure}[ht]
 \centerline{\centerline{\includegraphics[width=0.9  \textwidth]{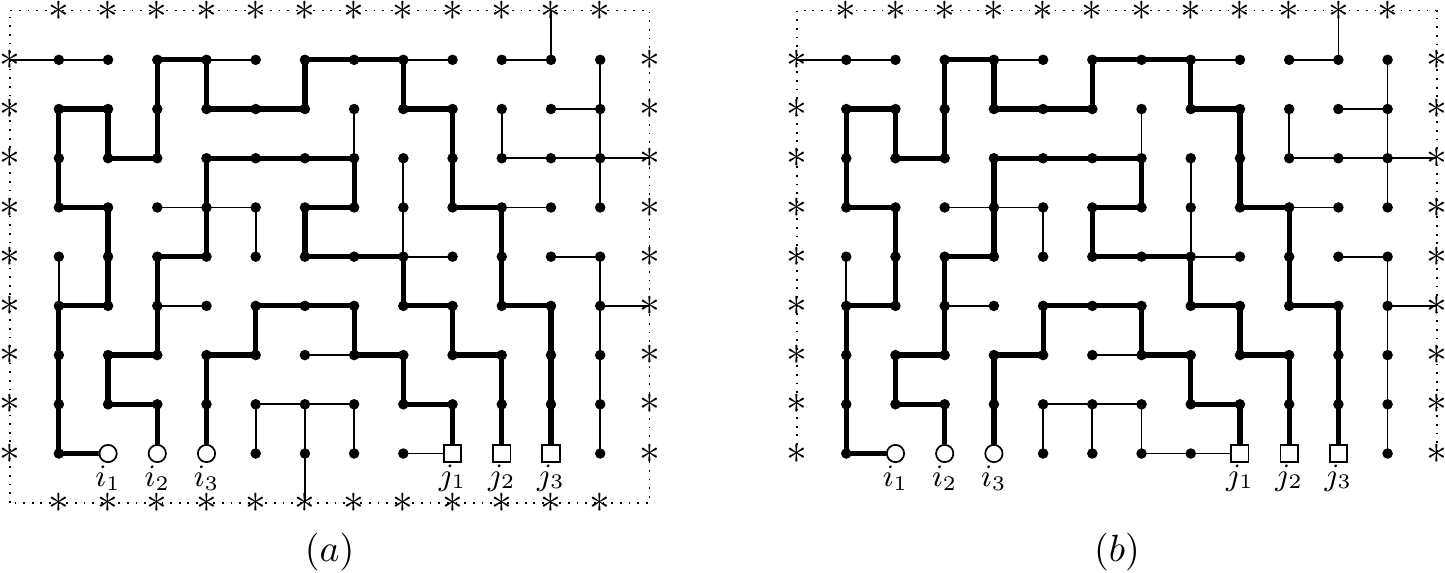}}}
 \caption{Watermelons with $k=3$ legs in  rectangular subsets $\mathcal{L}_n$ of the square  half-lattice with the open (a) and closed (b) BC at the lowest row and open boundary separating $\mathcal{L}_n$ from the rest of the half-lattice $\mathcal{L}$. The root sets $I=\{i_1,i_2,i_3\}$ ($\circ$) and endpoint set $J=\{j_1,j_2,j_3\}$ ($\openbox$) are strings of three sites in the lowest non-boundary row. The boundary sites are shown by asterisks (*) connected with the dotted line. Thick black lines show the corresponding LERWs.\label{fig: boundary watermelons}}
\end{figure}
 
In this way, we assign a probability to a subset of infinite SF configurations $\mathfrak{F}_{B}(IJ)$ containing $k$ paths  connecting sites of $I$ and $J$ inside a finite rectangular box $B$. Then, considering an exhaustion $B_1\subset B_2 \subset\dots\subset \mathcal{L}$ such that $\bigcup_{i\in\mathbb{N}}B_n=\mathcal{L}$, we define the probability of a watermelon on the infinite lattice as a limit
\[
 \mathbb{P}_{\mathcal{L}}^{I\cup\partial}(\mathfrak{F}(IJ|\partial)) =
 \lim\limits_{n\to\infty} \mathbb{P}_{\mathcal{L}}^{I\cup\partial}(\mathfrak{F}_{B_n}(IJ)).
\]
The latter limit exists, for the sequence is bounded and non-decreasing. Moreover, as it follows from the further explicit calculation, the limit is unique, i.e. does not depend on the exhaustion.

The standard approach would consist of two steps: first, to approximate the infinite lattice by a sequence of finite lattices, and second, to approximate an infinite  watermelon by a sequence of finite watermelons. However, instead, it is enough to consider a single diagonal sub-sequence which obviously converges to the same limit,    
\[
 \mathbb{P}_{\mathcal{L}}^{I\cup\partial}(\mathfrak{F}(IJ|\partial)) =
 \lim_{n\to\infty} \mathbb{P}_{\mathcal{L}_n}^{I\cup\partial_n}(\mathfrak{F}_{\mathcal{L}_n}(IJ|\partial_n)).
\]
This is the limit to be studied below.

Under the above definitions, the main result of the article is as follows.

\begin{theorem} \label{th: watermelon probabilities asymp}
 Consider the USF, i.e. $w(e)=1$ for any $e\in E$, on $\mathcal{L}^\mathrm{op}$ and $\mathcal{L}^\mathrm{cl}$ with additional $k$-site root-set $I$ and the set $J$ defined as in~\eref{eq: I,J}. Then, as $r\to\infty$, the probability of watermelon configurations asymptotically satisfies
 \begin{equation} \label{eq: open}
  \mathbb{P}_{\mathcal{L}^\mathrm{op}}^{I\cup\partial^\mathrm{op}}
   \big(\mathfrak{F}(IJ|\partial^\mathrm{op})\big) =
  C^\mathrm{op}_k\cdot r^{-k(k+1)}(1+o(1)),
 \end{equation}
 and  
 \begin{equation} \label{eq: closed}
  \mathbb{P}_{\mathcal{L}^\mathrm{cl}}^{I}
   \big(\mathfrak{F}(IJ)\big) =
  C^\mathrm{cl}_k\cdot r^{-k(k-1)}(1+o(1)),  
 \end{equation}
 respectively, where $C_k^\mathrm{op}$ and $C_k^\mathrm{cl}$ are constants defined below, see~\eref{eq: C_k^op} and~\eref{eq: C_k^cl}.
\end{theorem}

\subsection{Discussion of the results}

Let us first compare the exponents obtained with the predictions of CG theory and CFT that can be found in~\cite{Duplantier1986, DuplantierSaleur1986, DuplantierSaleur1987}.  
It is known, see \cite{Duplantier1989, Nienhuis1987} for review, that the $\Or(n)$  loop  model can be mapped to the SOS model that renormalizes into a Gaussian free field theory governed by the action  
\[
 A = \frac{g}{4\pi} \int (\partial\varphi)^2\rmd^2 x
\]
with coupling constant $g$ related to the loop weight
\[
 n = -2\cos \pi g.
\]
Here $g\in [1,2]$ corresponds to the critical point (dilute phase) of $\Or(n)$ loop model and $g\in[0,1]$ to the low-temperature regime (dense phase). Hence, the $L$-leg watermelon correlation function $G_L(\bi{r})$, i.e. suitably normalized partition function of loop configurations with $L$ polymers connecting two fixed endpoints separated by a~vector $\bi{r}$, asymptotically behaves as  
$G_L(\bi{r})\asymp |\bi{r}|^{-2x_L}$
when $|\bi{r}|$ is large. According to~\cite{Duplantier1986, DuplantierSaleur1986, DuplantierSaleur1987}, the critical exponent for the watermelon in the bulk of infinite system is
\[
 x_L^{b} = \frac{g L^2}{8}-\frac{(1-g)^2}{2g}.
\]
For the surface exponents describing the same correlation function near the boundary of the half-plane, there is a choice of exponents corresponding to different possible fixed points of the renormalization group.  We mention only the exponent 
\[
 x_L^{s} = \frac{g L^2}{4}+\frac{L(g-1)}{2}
\]
that describes the so-called ordinary phase transition  corresponding to the Dirichlet BC for the height field of the associated SOS model, which suggests that the polymers of the $\Or(n)$ model are reflected from the boundary~\cite{Jacobsen2009}. 

In these exponents, one can also recognize conformal weights associated with the CFT with central charge 
\[
 c = 1-\frac{6(1-g)^2}{g}
\]
given by the Kac formula 
\[
 h_{p,q} = \frac{((m+1)p-mq)^2-1}{4m(m+1)},
 \quad m\in\mathbb{N},
\]
where $p$ and $q$ are co-prime integers. For the dilute and dense loop phases, one finds
\[
 m = \frac{1}{g-1}, \quad
 x_L = 2h_{L/2,0}, \quad
 x_L^s = h_{L+1,1}
\]
and
\[
 m = \frac{g}{1-g}, \quad
 x_L = x_L=2h_{0,L/2}, \quad
 x_L^s = h_{1,L+1},
\]
respectively, though the appearance of half-integer indices in the bulk case is yet to be  understood.  

Our case is to be compared with the dense phase of $n=0$ limit of $\Or(n)$ corresponding to $g=1/2$ and   
\begin{equation}\label{eq: Polymer exponents}
 x_L^b = \frac{L^2}{16}-\frac{1}{4}, \quad
 x_L^s = \frac{L^2}{8}-\frac{L}{4}. 
\end{equation}

In order to compare these exponents with our results, we first note that ST or SF on the square lattice can be mapped to configurations of the $n=0$ version of $\Or(n)$ model, the dense  polymer model on the medial lattice~\cite{PearceRasmussen2007, BrankovGrigorevPriezzhevTipunin2008} that is the 45-degree rotated  square lattice with sites associated with bonds of  the original square lattice, see Figure~\ref{fig: polymers}. Under this mapping, every component of  SF  is surrounded by a loop or, equivalently, is embedded between a pair of polymers. Thus, a $k$-leg SF watermelon corresponds to the polymer watermelon with $L=2k$ legs. Indeed, the exponent~$x_{2k}^b$ is the one obtained in~\cite{Kenyon2000} and in~\cite{IvashkevichHu2005, GorskyNechaevPoghosyanPriezzhev2013}, where the power law was corrected with the logarithmic prefactor. 
\begin{figure}[ht]
 \centerline{\centerline{\includegraphics[width=0.9  \textwidth]{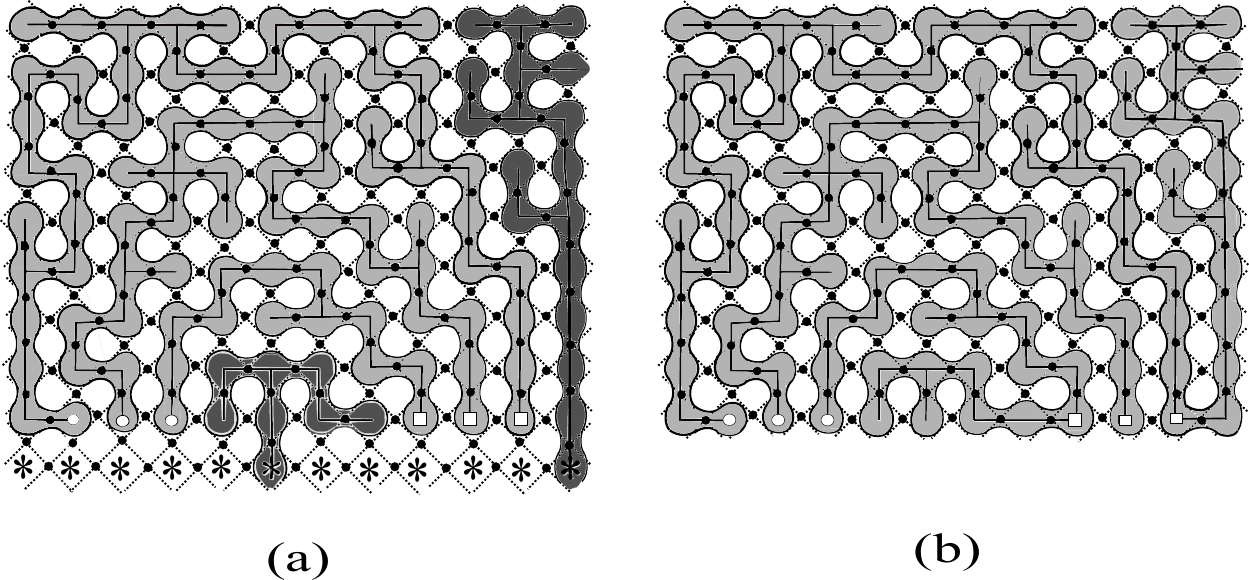}}}
 \caption{Correspondence between SF 3-leg watermelons on the subsets of semi-infinite square lattices  with open (a) and closed (b) boundary conditions and dense polymer configurations on the medial lattice. Every SF component is enveloped by the polymer loop. The area inside loops enveloping 3 watermelon components and connecting a root ($\circ$) with the corresponding endpoint ($\openbox$) is light gray shaded. Every such a loop is to be considered as a pair of polymers constituting 6-leg watermelon in the polymer picture. 
 In the closed BC case (b), the polymers are reflected from the boundary. Thus,  the value of the height function assigned to faces of the medial  lattice, which  changes by one on every polymer crossing,  stays constant along the boundary maintaining the Dirichlet BC.    In the case of open boundary (a), other (dark gray shaded) components   rooted to boundary roots ($*$)  may exist. Corresponding polymers go away from the lattice causing varying the height function along the boundary.  }
 \label{fig: polymers}	
\end{figure}  

To identify the surface exponent $x_{2k}^s$ with the  exponents obtained, we note that the Dirichlet BC for the height field,\footnote{ To define the height function assigning integer values to faces of the medial lattice, we first consider an oriented polymer configuration by giving an orientation to the loops. In this case, the directed loop configuration determines the height function by the condition that the value of the function increases (decreases) by one whenever we cross left- (right-) oriented polymer, when going between two neighboring faces. Given a SF on a restricted domain with closed BC, we obtain a loop configuration, were polymers are reflected from the boundary (Figure~\ref{fig: polymers}b). If we think of the domain as a part of an infinite lattice, then the corresponding height function can be consistently fixed to be constant outside the domain, i.e. satisfying Dirichlet BC, since no polymer is crossed when going around the domain. This is not true for loop configurations obtained from SF with open BC, where presence of extra  rooted SF components results in sources and sinks of polymers at the boundary (Figure~\ref{fig: polymers}a). See~\cite{Jacobsen2009} for details.}
which suggest that the polymers are reflected from the boundary~\cite{Jacobsen2009}, correspond to the closed (Neumann) BC for SF. Indeed, substituting $L=2k$ to $x_{L}^s$ from~\eref{eq: Polymer exponents}, we obtain power law~\eref{eq: closed}.


To explain the origin of the  exponent in~\eref{eq: open} for the  the half-lattice with open BC, we note that  SF on a finite domain of the square lattices  with open and closed BC are dual to each other, see e.g. \cite{KarrilaKytolaPeltola2020}. More specifically, consider  a simply connected finite domain $G$ of the square lattice with closed BC and its dual domain $G'$ of the dual lattice with open BC, so that the sites of $G'$  are associated with faces of $G$  including  site $(*)$ called  the root (open boundary) associated with the  external face of $G$. Then,  the set of unrooted  ST on $G$   is in bijection with the set of dual ST rooted to $(*)$. This correspondence can obviously be promoted to that between SF sets  on the half-infinite lattices $\mathcal{L}^{\mathrm{op}}$ and $\mathcal{L}^{\mathrm{cl}}$. 
Furthermore, one can see that a similar bijection holds between the sets $\mathfrak{F}(I_kJ_k|\partial^\mathrm{op})$ and  $\mathfrak{F}(I_{k+1}J_{k+1})$  of $k$-watermelons on  $\mathcal{L}^{\mathrm{op}}$ and $(k+1)$-watermelons on $\mathcal{L}^{\mathrm{cl}}$, respectively, where the subscripts show the cardinality of sets $|I_k|=|J_k|=k$ (to see this, look at the light gray areas at Figure~\ref{fig: polymers} and their complementary white areas). This is the reason, why the exponent of~\eref{eq: open} can be obtained from that of~\eref{eq: closed} by shift $k\to (k+1)$ as well as the numerator of $C_{k+1}^{\mathrm{op}}$ (see formulas (\ref{eq: C_k^op}) and (\ref{eq: C_k^cl}) below). Note that there is no such a duality between the whole sets $\mathfrak{F}(I_k|\partial^\mathrm{op})$ and  $\mathfrak{F}(I_{k+1})$ of~SF with~$k$ and $(k+1)$  additional components, respectively. Hence, the normalization factors contributing to the denominators of $C_{k+1}^{\mathrm{op}}$ and $C_{k}^{\mathrm{cl}}$ are still different.

Our result can also be compared with the result of Kenyon~\cite{Kenyon2000} for the crossing exponent, which defines the crossing probability decay rate on the rectangle. Under the conformal map of the rectangle to the semi-annulus, the exponential decay of~\cite{Kenyon2000} becomes the power law~\eref{eq: closed}, see also~\cite{Cardy1984} for the relation between critical exponents and amplitude of the correlation length in finite size systems.

\section{Kirchhoff theorem and Green function}\label{sec: Method of solution in general}

\subsection{Matrix tree theorem}
In this section, we describe a general approach for evaluating the watermelon probability mainly following~\cite{Poncelet2018}. Let us consider a finite directed connected graph $\mathcal{G}=(V,E)$  without self-loops and multiple edges. Let $N$ be the number of its non-boundary vertices, $N=|V\setminus\partial|$, and $\Delta=(\Delta_{ij})_{i,j\in V\setminus\partial}$ be an $N\times N$ matrix whose elements are defined by the formula
\begin{equation}\label{formula: graph matrix}
	\Delta_{ij}=\left\{\begin{array}{cl}
		\sum\limits_{k\ne j}w(jk) & \quad\mbox{if} \;\; i=j\\
		-w(ji) & \quad\mbox{if} \;\; i\ne{j}.
	\end{array}\right.
\end{equation}
Here we suppose that the weight $w(ij)=0$ if $ij\notin E$. Thus, $-\Delta_{ij}$ is a weight of a~directed edge from site $j$ to site $i$ for different values of $i$ and $j$, and $\Delta_{ii}$ is the sum of the weights of directed edges from site~$i$. In particular, the sum of $j$-th column elements is equal to the sum of the weights of  edges directed from $j$ to $\partial$,
\[
 \sum\limits_{i=1}^{N}\Delta_{ij} = 
 \sum\limits_{k\in\partial}w(jk).
\]
Note that directed graphs with multiple edges can be considered in the same frame if one puts $w(ij)$ to be the sum of all weights of directed edges from $i$ to $j$. Undirected graphs suit this frame as well: it is sufficient to take $w(ij)=w(ji)$ for all $i\ne j$.

Recall the Matrix Tree Theorem which, being stated in a slightly different way, was first proved by Kirchhoff.

\begin{theorem}[Matrix Tree Theorem,~\cite{Kirchhoff1847}]\label{thm: Kirchhoff}
Let the graph $\mathcal{G}$ with boundary $\partial$ be as defined above, and let $\Delta$ be the matrix of its discrete Laplacian. Then we have
 \[
  Z_{\mathcal{G}}(\partial) = \det\Delta.
 \]
\end{theorem}

The main tool for our investigation is the All Minors Matrix Tree Theorem which is a generalization of the above theorem. For  $I\subset V \backslash \partial$, denote $\bar{I}=V\setminus (I\cup\partial)$ and define $(-1)^{\Sigma I}$ to be the sum of indices within $I$ with respect to a linear order of indices within $V$. Additionally, for $I,J\subset V\backslash \partial$, define $\Delta^J_I$ to be the restriction of matrix $\Delta$ to the rows indexed by the vertices of $J$ and to the columns indexed by the vertices of $I$. Also, for a permutation $\sigma\in S_k$, denote
\[
 \mathfrak{F}(I\sigma(J)|\partial) \equiv
 \mathfrak{F}(i_1j_{\sigma(1)}|\dots|i_kj_{\sigma(k)}|\partial)
\]
the set of forests such that for any $\mathcal{F}\in\mathfrak{F}(I\sigma(J)|\partial)$:
\begin{itemize}
 \item
  every component of $\mathcal{F}$ is rooted to $I\cup\partial$,
 \item
  for every $l=1,\ldots,k$, the site $j_{\sigma(l)}$ belongs to the component rooted to $i_l$.
\end{itemize}
Then the following theorem holds.

\begin{theorem}[All Minor Matrix Tree Theorem,~\cite{Chen1976, Chaiken1982}]\label{theorem: All Minor Matrix Tree Theorem}
	Let $|I|=|J|=k$. Then
\begin{equation}
	\det\Delta^{\bar{J}}_{\bar{I}} = (-1)^{\Sigma I+\Sigma J}\sum\limits_{\sigma\in S_{k}}(-1)^{|\sigma|}Z_{\mathcal{G}}(I\sigma(J)|\partial),
	\label{eq: All Minor Matrix Tree Theorem}
	\end{equation}
	where the sum runs over the symmetric group on a set of  $k$ elements, and $Z_{\mathcal{G}}(I\sigma(J)|\partial)$ is the partition function of $\mathfrak{F}(I\sigma(J)|\partial)$.	
\end{theorem}

When the matrix $\Delta$ is invertible, which is always the case by Theorem \ref{thm: Kirchhoff} for $\partial\neq \emptyset$, in addition to Theorem~\ref{theorem: All Minor Matrix Tree Theorem} we can use the following formula.

\begin{proposition}[Jacobi's complementary minor formula,~\cite{Gantmacher1959}]\label{prop: Jacobi's complementary minor formula}
 Let $\Delta$ be invertible, with  the Green function  $G=\Delta^{-1}$. Then 
 \[
  \det G^I_J = (-1)^{\Sigma I+\Sigma J}\cdot\frac{\det\Delta^{\bar{J}}_{\bar{I}}}{\det\Delta}.
 \]
\end{proposition}

Using Theorem~\ref{theorem: All Minor Matrix Tree Theorem} and Proposition~\ref{prop: Jacobi's complementary minor formula} we can finally write the principle formula of our interest. To this end, we note that we work  with finite connected undirected graphs $\mathcal{L}_n=\mathcal{L}_n^\mathrm{op} \subset \mathcal{L}^\mathrm{op}$  
(respectively, $\mathcal{L}_n^\mathrm{cl} \subset \mathcal{L}^\mathrm{cl}$) with the boundary $\partial_n^\mathrm{op}$ (respectively, $\partial_n^\mathrm{cl}$) and   the sets $I$ and $J$ of sites defined in~\eref{eq: I,J}. In both cases, the only non-zero summand of the sum in r.h.s of~\eref{eq: All Minor Matrix Tree Theorem} corresponds to $\sigma=\mathrm{Id}$. Hence, according to~\eref{eq: watermelonprob}, the probability of the watermelon is
\begin{equation}\label{eq: watermelon probability in general}
 \mathbb{P}_{\mathcal{L}_n}^{I\cup\partial_n}\big(\mathfrak{F}_{\mathcal{L}_n}(IJ|\partial)\big)=\frac{Z_{\mathcal{L}_n}(IJ|\partial_n)}{Z_{\mathcal{L}_n}(I|\partial_n)}=\frac{\det (G^n)^{I}_{J}}{\det (G^n)^{I}_{I}}, 
\end{equation}
where $G^n$ is the Green function associated with $\mathcal{L}_n$. Note that, for arbitrary $\mathcal{L}_n$, the explicit formulas of both the numerator and denominator of the fraction in r.h.s. of~\eref{eq: watermelon probability in general} do not exist. However, the local $n\to\infty$ behavior of the Green functions is well known and will be used below to evaluate the $n\to\infty$ limit of the l.h.s. of~\eref{eq: watermelon probability in general}.

\subsection{Green function}
The Green function $G^{n}_{\mathbf{x_1},\mathbf{x_2}}$ associated with the lattice $\mathcal{L}_n$ is the inverse matrix of the corresponding discrete Laplacian, i.e. it satisfies
\begin{equation}\label{formula: equation for G_(x,y)}
 4G^n_{\mathbf{x_1},\mathbf{x_2}} -
  \big(G^n_{\mathbf{x_1}-\mathbf{e}_x,\mathbf{x_2}} +
   	   G^n_{\mathbf{x_1}+\mathbf{e}_x,\mathbf{x_2}} +
       G^n_{\mathbf{x_1}-\mathbf{e}_y,\mathbf{x_2}} +
       G^n_{\mathbf{x_1}+\mathbf{e}_y,\mathbf{x_2}}
  \big) =
 \delta_{\mathbf{x_1},\mathbf{x_2}}
\end{equation}
for any $\mathbf{x_1},\mathbf{x_2}\in \mathcal{L}_n\backslash \partial_n$, supplied with Dirichlet boundary conditions 
\[
 G^{n}_{\mathbf{x_1},\mathbf{x_2}} = 0, \quad
 \forall\,\mathbf{x_1}\in \partial_n,\,
 \forall\,\mathbf{x_2}\in\mathcal{L}_n,
\]
where $\mathbf{e}_x=(1,0)$ and $\mathbf{e}_y=(0,1)$ are the basis  vectors of the lattice. In particular, when $\mathcal{L}_n=\mathcal{L}_n^{\mathrm{op}}$, the Dirichlet boundary conditions on the lowest row of $\mathcal{L}^{\mathrm{op}}$ suggest  
\[
 G^{\mathrm{op},n}_{\mathbf{x_1},\mathbf{x_2}} = 0, \quad
 \forall\,\mathbf{x_1}\in \{(k,0)\}_{k \in \mathbb{Z}},\,
 \forall\,\mathbf{x_2}\in\mathcal{L}_n.
\]
Alternatively, when $\mathcal{L}_n=\mathcal{L}_n^{\mathrm{cl}}$, we impose Neumann boundary conditions on the lowest row, i.e. 
\[
 G^{\mathrm{cl},n}_{\mathbf{x_1},\mathbf{x_2}} = 
 G^{\mathrm{cl},n}_{\mathbf{x_1}-\mathbf{e}_y,\mathbf{x_2}}, \quad
 \forall\,\mathbf{x_1}\in \{(k,1)\}_{k \in \mathbb{Z}},\,
 \forall\,\mathbf{x_2}\in\mathcal{L}_n.
\]

The Green function for the lattice $\mathcal{L}_n$ can be obtained by the image method from the Green function for the lattice $\mathcal{L}_n^{\mathrm{sym}}=\mathcal{L}_n\cup{\mathcal{L}}_n^*$ contained in the infinite lattice~$\mathcal{L}$ with the set of vertices $\mathbb{Z}^2$ extending the original half-lattice. If $\mathcal{L}_n=\mathcal{L}^{\mathrm{op}}$, then the vertex set of the half-lattice is $\mathbb{Z}\times \mathbb{Z}_{\geqslant 0}$ and the additional part $\mathcal{L}^*_n=(V,E)$ is obtained by reflection with respect to the horizontal line $y=0$, meaning that
 \[
  V = \{\mathbf{x}^*=(x,-y):\, \mathbf{x}=(x,y)\in \mathcal{L}_n\}.
 \]
In the case when $\mathcal{L}_n=\mathcal{L}^{\mathrm{cl}}$, the vertex set of the half-lattice is $\mathbb{Z}\times \mathbb{Z}_{>0}$ and the reflection is done with respect to the line $y=1/2$, i.e.
 \[
  V = \{\mathbf{x}^*+\mathbf{e}_y=(x,1-y):\,
      \mathbf{x}=(x,y)\in \mathcal{L}_n\}.
 \]
We also imply that sites  from the lowest row  $\partial^{\mathrm{op}}$, which are boundary sites in $\mathcal{L}_n^{\mathrm{op}}$, are not boundary in the corresponding lattice $\mathcal{L}_n^{\mathrm{sym}}$ anymore, while the other boundary sites from the bulk of $\mathcal{L}^{\mathrm{op}}$ or $\mathcal{L}^{\mathrm{cl}}$, as well as their mirror images, are. An important feature of $\{\mathcal{L}_n^{\mathrm{sym}}\}_{n\in\mathbb{N}}$ is that it is an exhausting sequence of the infinite lattice $\mathcal{L}$. In particular, the boundary of its entries goes away to infinity, as $n\to\infty$.

Using the notation $G^{\mathrm{sym},n}_{\mathbf{x_1},\mathbf{x_2}}$ for the Green function in the symmetric domain $\{\mathcal{L}_n^{\mathrm{sym}}\}_{n\in\mathbb{N}}$, we obtain the Green function for the original subsets,
\begin{eqnarray}
 G^{\mathrm{op},n}_{\mathbf{x_1},\mathbf{x_2}} &=&
 G^{\mathrm{sym},n}_{\mathbf{x_1},\mathbf{x_2}} -
 G^{\mathrm{sym},n}_{\mathbf{x_1},\mathbf{{x}^*_2}},
 	\label{eq: G^(op)_(x,y) through G_(x,y)}\\
 G^{\mathrm{cl},n}_{\mathbf{x_1},\mathbf{x_2}} &=&
 G^{\mathrm{sym},n}_{\mathbf{x_1},\mathbf{x_2}} +
 G^{\mathrm{sym},n}_{\mathbf{x_1},\mathbf{{x}^*_2} +
 \mathbf{e}_y},
 	\label{eq: G^(cl)_(x,y) through G_(x,y)}
\end{eqnarray}
where ${\mathbf{x}^*}=(x,-y)$ for $\mathbf{x}=(x,y)$. Relations~\eref{eq: G^(op)_(x,y) through G_(x,y)} and~\eref{eq: G^(cl)_(x,y) through G_(x,y)} follow from the fact that the second summands in both equations are harmonic functions in $\mathcal{L}^{\mathrm{op}}$ and $\mathcal{L}^{\mathrm{cl}}$, respectively, while the sums manifestly satisfy boundary conditions.  

The Green function has a transparent meaning in the languages of electric circuits and random walks. If we consider the edges of the graph 
as one Ohm resistors, then the value of $G_{\mathbf{x_1},\mathbf{x_2}}$ represents the voltage at the site $\mathbf{x_1}$, given a unit current is injected into the site $\mathbf{x_2}$, while the Dirichlet boundary is grounded.  The value of $G_{\mathbf{x_1},\mathbf{x_2}}$ can also be interpreted as an expected number of visits of the site $\mathbf{x_1}$ before leaving the boundary by the random walk started at $\mathbf{x_2}$.

Given an exhausting sequence of subgraphs of a periodic planar lattice,  the voltage necessary to maintain the unit current from a fixed site to the boundary grows unboundedly, which is a consequence of the recurrence of the random walk in two-dimensions. In particular, this is the case for  $\{\mathcal{L}_n^{\mathrm{sym}}\}_{n\in\mathbb{N}}$, which means that $G^{\mathrm{sym},n}_{\mathbf{x_1},\mathbf{x_2}}\to \infty$, as $n\to \infty$, for any fixed sites $\mathbf{x_1,x_2}$.
 
At the same time, the limiting voltage drop $g_{\mathbf{x_1-x_2}}$ between two sites at a finite distance from each other is finite and well defined, being an increase of the potential kernel of the corresponding random walk. In particular~\cite{Spitzer2001}, for $\mathbf{x}_i=(x_i,y_i)$ with $i=1,2$, we have
 \[
  g_{\mathbf{x_1-x_2}} =
  \lim_{n\to\infty}
   (G^{\mathrm{sym},n}_{\mathbf{x_1},\mathbf{x_2}} - 
   G^{\mathrm{sym},n}_{\mathbf{x_2},\mathbf{x_2}})
 \]
 \begin{equation}\label{formula: integral for G_(x,y)}
  {\color{white} g_{\mathbf{x_1-x_2}}} =
  \frac{1}{2\pi^2} \int\limits_{0}^{\pi}\rmd\alpha \int\limits_{0}^{\pi}\rmd\beta\, \frac{\cos{(x_1-x_2)\alpha}\cdot\cos{(y_1-y_2)\beta}-1} {2-(\cos\alpha+\cos\beta)}.
 \end{equation}
Roughly speaking, at   the infinite lattice the   Green function $G_{\mathbf{x_1},\mathbf{x_2}}$  can be thought of as  a sum of the  infinite term $G_{\mathbf{x_2},\mathbf{x_2}}$, which is the voltage at the site $\mathbf{x_2}$ of current injection or the expected number of returns of the random walk to the origin $\mathbf{x_2}$, and the well-defined finite voltage drop $g_{\mathbf{x_1-x_2}}$ between sites
$\mathbf{x_1}$ and $\mathbf{x_2}$, which depends only on the difference $(\mathbf{x_1}-\mathbf{x_2})$ of the lattice coordinates. 

Also, a fact important for further derivation is the translation invariance of the infinite part of $G_{\mathbf{x},\mathbf{x}}$, i.e. its independence of $\mathbf{x}$. Due to the translation invariance of the infinite square lattice, it may seem self-evident and often referred to as such in the literature. On the other hand, since $G_{\mathbf{x},\mathbf{x}}$ is infinite, the statement should be formulated in terms of the sequences $\{G^n_{\mathbf{x},\mathbf{x}}\}_{n\in\mathbb{N}}$ of Green functions associated with the finite subsets of the infinite lattice. 

\begin{lemma}\label{lem: electric}
 Let $\mathcal{L}=({V,E})$ with 
 \[
  V = \mathbb{Z}^2, \quad
  E = \{(v,v+\mathbf{e}_x),(v,v+\mathbf{e}_y)\}_{v\in V },
 \]
 and $\{\mathcal{L}_n\}_{n\in\mathbb{N}}$ be an exhausting sequence of connected lattice subsets. Assume that each subset~$\mathcal{L}_n$ has a Dirichlet boundary $\partial_n$ consisting of the sites adjacent to the sites of~$\mathcal{L}$ outside of $\mathcal{L}_n$. Let $\mathbf{x}$ and $\mathbf{y}$ be two fixed sites of $\mathcal{L}$ such that $\mathbf{x,y}\in \mathcal{L}_n$ for any $n\in \mathbb{N}$. Then 
 \[
  \lim_{n\to\infty} |G^n_{\mathbf{x,x}} - G^n_{\mathbf{y,y}}|=0.
 \]
\end{lemma}
\begin{proof}
 The proof is based on the electric interpretation of the Green function and the Rayleigh's monotonicity principle~\cite{DoyleSnell1984} that suggests that if the resistances of a circuit are increased (respectively, decreased), then the effective resistance between any two sites cannot decrease (respectively, increase). For a given integer $n$, let us consider an auxiliary lattice subset 
 \[
  \mathcal{L}'_n = T_{\mathbf{x-y}} \mathcal{L}_n \subset \mathcal{L}
 \]
 obtained from $\mathcal{L}_n$ with the help of translation $T_{\mathbf{x-y}}$ by a vector $(\mathbf{x-y})$ that sends the site $\mathbf{y}$ to the site $\mathbf{x}$. Define also
 \[
  \hat{\mathcal{L}}_n = \mathcal{L}_n \cap \mathcal{L}'_n  
  \qquad\mbox{and}\qquad 
  \check{\mathcal{L}}_n = \mathcal{L}_n \cup \mathcal{L}'_n.
 \]
 The Dirichlet boundaries $\partial'_n$, $\hat{\partial}_n$ and $\check{\partial}_n$ corresponding to the subsets $\mathcal{L}'_n$, $\hat{\mathcal{L}}_n$ and $\check{\mathcal{L}}_n$, respectively, consist of the sites connected to sites of $\mathcal{L}$ beyond these sets (see Figure~\ref{fig: domains}~(a) and~(b)). 
 \begin{figure}[ht]
  \centerline{\includegraphics[width=0.95\textwidth]{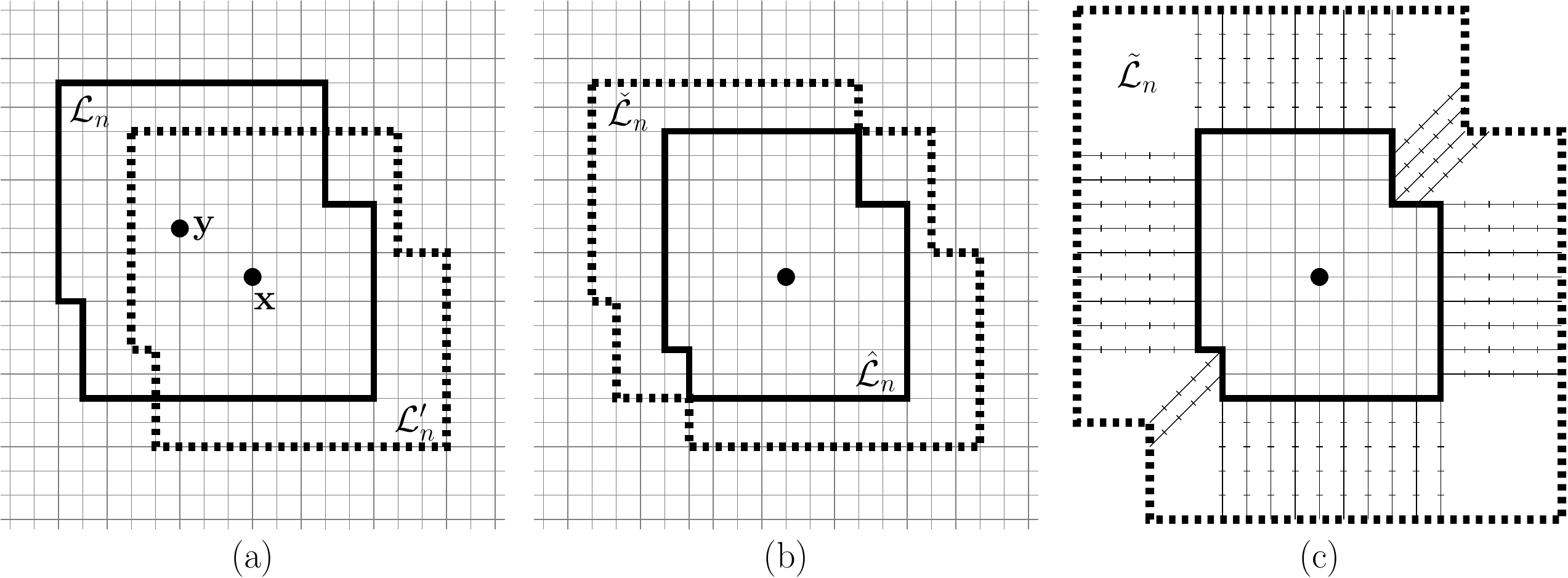}}
  \caption{The graphs used in the proof of Lemma~\ref{lem: electric}: (a) the subsets $\mathcal{L}_n$ and $\mathcal{L}'_n$ with the boundaries~$\partial_n$ and $\partial'_n$ consisting of sites on the solid and dashed contour, respectively; (b) the subsets $\hat{\mathcal{L}}_n$ and $\check{\mathcal{L}}_n$ with the boundaries $\hat{\partial}_n$ and $\check{\partial}_n$ consisting of sites on the solid and dashed contour, respectively; (c) the subset $\tilde{\mathcal{L}}_n$ with the boundary~$\tilde{\partial}_n$ consisting of sites on the dashed contour. The sites of $\tilde{\partial}_n$ are connected to the sites of $\hat{\mathcal{L}}_n$ by the strings of $|\mathbf{x}-\mathbf{y}|=5$ bonds. }
  \label{fig: domains}	
 \end{figure}

 Suppose that $G^n_{\mathbf{x,x}}>G^n_{\mathbf{y,y}}$. Since $\hat{\mathcal{L}}_n \subset \mathcal{L}_n$ and $\mathcal{L}'_n\subset \check{\mathcal{L}}_n$, according to Rayleigh's principle, we have
 \begin{equation}\label{eq: Green functions inequality}
  G^n_{\mathbf{x,x}} - G^n_{\mathbf{y,y}} \leqslant
  \check{G}^{n}_{\mathbf{x,x}} - \hat{G}^{n}_{\mathbf{x,x}}, 
\end{equation} 
where $\check{G}^{n}_{\mathbf{x,x}}$ and $\hat{G}^{n}_{\mathbf{x,x}}$ are the Green functions of the corresponding sets. To estimate the difference in  r.h.s. of~\eref{eq: Green functions inequality}, let us bound $\check{G}^{n}_{\mathbf{x,x}}$ as follows. Consider a new graph $\tilde{\mathcal{L}}_n$ whose internal vertices coincide with those of $\hat{\mathcal{L}}_n$ (including $\hat{\partial}_n$), such that every site of $\hat{\partial}_n$ is connected to the ground, i.e. a new Dirichlet boundary $\tilde{\partial}_n$, by a resistor with resistance~$|\mathbf{x}-\mathbf{y}|$ equal to the lattice distance between $\mathbf{x}$ and $\mathbf{y}$ (see Figure~\ref{fig: domains}~(c)). Since the lattice distance from every site of $\hat{\partial}_n$ to the nearest site of $\check{\partial}_n$ is at most $|\mathbf{x}-\mathbf{y}|$,  the effective resistance from $\mathbf{x}$ to $ \tilde{\partial}_n$ in $ \tilde{\mathcal{L}}_n$ exceeds that from $\mathbf{x}$ to $\check{\partial}_n$  in $\check{\mathcal{L}}_n$, i.e.
 \[
  \check{G}^{n}_{\mathbf{x,x}}\leqslant \tilde{G}^{ n}_{\mathbf{x,x}}.
 \]
 The value of $\tilde{G}^{ n}_{\mathbf{z,x}}$ at all non-boundary vertices $\mathbf{z}$ of $\tilde{\mathcal{L}}_n$ can be obtained by shifting the voltages $\hat{G}^{n}_{\mathbf{z,x}}$ by a constant value equal to the  voltage drop on the resistors with the resistance $|\mathbf{x}-\mathbf{y}|$ connecting  sites of $\hat{\partial}_n$ to $\tilde{\partial}_n$. The shift value can be found from the fact that the unit current flowing through these resistors to the ground is equally distributed among the number $|\hat{\partial}_n|$ of them. Thus,
 \begin{equation}\label{eq: Green functions inequality - 2}
  G^n_{\mathbf{x,x}} - G^n_{\mathbf{y,y}} \leqslant 
  \tilde{G}^{n}_{\mathbf{x,x}} - \hat{G}^{n}_{\mathbf{x,x}} =
  |\mathbf{x}-\mathbf{y}|/|\hat{\partial}_n|.
 \end{equation}   
The statement follows from the fact that for large enough $n$ the ratio in r.h.s of~\eref{eq: Green functions inequality - 2} is arbitrarily small. 
\end{proof}

The consequence of the above lemma is that in calculations we can use the infinite lattice Green function in the form 
\[
 G_{\mathbf{x_1,x_2}} = G_{0,0}+g_{{\mathbf{x_1-x_2}}}.
\]
Here the coordinate independent infinite part $G_{0,0}$ should be understood as
 \[
  G_{0,0} = G^n_{(\mathbf{x},\mathbf{x})}+o(1),
 \]
as $n\to\infty$, and, for $\mathbf{x}_i=(x_i,y_i)$ with $i=1,2$, the function  $g_{\mathbf{x_1-x_2}}=g_{(x_1-x_2,y_1-y_2)}$ is the finite part given by the r.h.s. of~\eref{formula: integral for G_(x,y)}, which depends on the relative coordinates only. Hence, as $n\to\infty$, it follows from~\eref{eq: G^(op)_(x,y) through G_(x,y)} and~\eref{eq: G^(cl)_(x,y) through G_(x,y)} that 
\begin{eqnarray}
 G^{\mathrm{op}}_{\mathbf{x_1},\mathbf{x_2}} & = &
 G^{\mathrm{op}}_{(x_1-x_2;y_1,y_2)} =
 g_{(x_1-x_2,y_1-y_2)} - g_{(x_1-x_2,y_1+y_2)},
 	\label{eq: G^(op)_(x,y) through g_(x-y)}  \\  
 G^{\mathrm{cl}}_{\mathbf{x_1},\mathbf{x_2}} & = &
 G^{\mathrm{cl}}_{(x_1-x_2;y_1,y_2)} =
 2G_{0,0} + g_{(x_1-x_2,y_1-y_2)} + g_{(x_1-x_2,y_1+y_2-1)},
    \label{eq: G^(cl)_(x,y) through g_(x-y)} 
\end{eqnarray}
where, for further brevity, we introduce notations 
\begin{equation}\label{formula: integral for G^(op)_(x,y)}
 G^{\mathrm{op}}_{(x;y_1,y_2)} = \frac{1}{\pi^2} \int\limits_{0}^{\pi}\rmd\alpha \int\limits_{0}^{\pi}\rmd\beta\, \frac{\cos{x\alpha}\cdot\sin{y_1\beta}\cdot\sin{y_2\beta}} {2-(\cos\alpha+\cos\beta)},
\end{equation}
\begin{equation}\label{formula: integral for G^(cl)_(x,y)}
 G^{\mathrm{cl}}_{(x;y_1,y_2)} = 2G_{0,0} 
\end{equation}
\[
 {\color{white}G^{\mathrm{cl}}_{(x;y_1,y_2)}}
  + \frac{1}{\pi^2} \int\limits_{0}^{\pi}\rmd\alpha \int\limits_{0}^{\pi}\rmd\beta\, \frac{\cos{x\alpha}\cdot\cos\big(y_1-1/2\big)\beta\cdot\cos\big(y_2-1/2\big)\beta-1} {2-(\cos\alpha+\cos\beta)}
\]
in which   the translation invariance, i.e. dependence on relative coordinates, in the direction parallel to the boundary is  explicitly incorporated. One can see that, similarly to the infinite lattice case, the Green function at the semi-infinite half-lattice with closed boundary expectedly has an infinite part, which will be crucial for the pure power law asymptotics of the watermelon probability.

In this article, the asymptotic large-distance behavior of the Green function is important. In particular, the asymptotics of the finite part of the infinite lattice Green function, as the distance $r$ along the horizontal direction between the sites with the same vertical coordinate grows to infinity, reads~\cite{Spitzer2001} 
\begin{equation}\label{formula: approximation for a_(m,n)}
 g_{(r,0)} =
  -\frac{1}{2\pi}\ln{r} -
   \frac{1}{\pi}\left(\frac{\gamma}{2}+\frac{3}{4}\ln2\right) +
   \frac{1}{24\pi{r}^2} +
   O\left(\frac{1}{r^4}\right).
\end{equation}
Hence, using relations~\eref{eq: G^(op)_(x,y) through g_(x-y)} and~\eref{eq: G^(cl)_(x,y) through g_(x-y)}, the symmetry $g_{(m,n)}=g_{(\pm n,\pm m)}$ and the recurrent formula
\[
 4g_{(m,n)} = g_{(m+1,n)} + g_{(m,n+1)} + g_{(m-1,n)} + g_{(m,n-1)}
\]
for $(m,n)\neq(0,0)$, we can express the asymptotic behavior of functions $G^{\mathrm{op}}_{(r;1,1)}$ and~$G^{\mathrm{cl}}_{(r;1,1)}$ involved in further calculations as follows:
\begin{eqnarray}
 G^{\mathrm{op}}_{(r;1,1)} &=& \frac{1}{\pi{r}^2}-O\left(\frac{1}{r^4}\right),
 	\label{eq: approximation for G^(op)_(r;1,1)} \\
 G^{\mathrm{cl}}_{(r;1,1)} &=& 2G_{0,0} -\frac{\ln{r}}{\pi}+O\left(1\right).
 	\label{eq: approximation for G^(cl)_(r;1,1)}
\end{eqnarray}

\section{Evaluating the determinants}\label{sec: Counting the determinants}
As we saw in the previous section, the probability of watermelon connecting the strings $I$ and $J$ of sites near the boundary of the half-lattice is given by the ratio of determinants $\det{G}_I^J$ and $\det{G}_I^I$, where $G$ is ether $G^{\mathrm{op}}$ or $G^{\mathrm{cl}}$ for open and closed boundary, respectively. This is why the main goal of this section is evaluation of leading asymptotics of special determinants. It is worth mentioning that we obtain more general result than we will actually need for our purpose.

Throughout this section, we use the following notations. Given a positive integer~$k$, we consider two $k$-tuples of variables $\mathbf{v}=(v_1,\dots,v_k)$ and $\mathbf{u}=(u_1,\dots,u_k)$. We also consider a formal power series $f(t)$ in one variable,
\begin{equation}\label{eq: f(t)}
 f(t) = \sum\limits_{l=0}^{\infty}b_l{t}^l,
\end{equation}
with the sequence of coefficients $\mathbf{b}=\{b_i\}_{i\in \mathbb{N}_0}$. Our goal is to evaluate the determinant
\begin{equation}\label{eq: F_k(u,v;b)}
 F_k(\mathbf{u},\mathbf{v};\mathbf{b}) =
 \det\limits_{1\leqslant{i,j}\leqslant{k}}\big[f(v_i-u_j)\big].
\end{equation}
Its behavior is given by the following fundamental lemma.

\begin{lemma}\label{lemma: fundamental determinant lemma}
 If $F_k(\mathbf{u},\mathbf{v};\mathbf{b})$ is given by~\eref{eq: F_k(u,v;b)}, then
\begin{equation}\label{eq: F_k statement}
 F_k(\mathbf{u},\mathbf{v};\mathbf{b}) =
 \Delta(\mathbf{v}) \Delta(-\mathbf{u})
 \sum\limits_{\lambda,\mu\in\mathcal{I}_k}
  C_{\lambda,\mu} s_{\lambda}(\mathbf{v}) s_{\mu}(-\mathbf{u}),
\end{equation}
 where
 \begin{itemize}
  \item the summation indices $\mu=(\mu_1\geqslant\cdots\geqslant\mu_k\geqslant 0)$ and $\lambda=(\lambda_1\geqslant\cdots\geqslant\lambda_k\geqslant 0)$ run over the set $\mathcal{I}_k$ of partitions with at most $k$ non-zero parts,
  \item $\Delta(\mathbf{v})$ and $\Delta(\mathbf{-u})$ are the Vandermonde determinants, i.e.
  \[
   \Delta (\mathbf{v})=\prod_{1\leqslant i<j\leqslant k}(v_i-v_j)
   \qquad\mbox{and}\qquad
   \Delta (-\mathbf{u})=\prod_{1\leqslant i<j\leqslant k}(-u_i+u_j),
  \]
  \item $s_{\lambda}(\mathbf{v})=s_{\lambda}(v_1,\ldots,v_k)$ and $s_{\mu}(-\mathbf{u})=s_{\mu}(-u_1,\ldots,-u_k)$ are Schur symmetric polynomials,
  \item the constants $C_{\lambda,\mu}$ are given by
   \[
	C_{\lambda,\mu} =
	\det\limits_{1\leqslant{i,j}\leqslant{k}}
	 \left[b_{\lambda_j+\delta_j+\mu_i+\delta_i}
	  {\lambda_j+\delta_j+\mu_i+\delta_i \choose \lambda_j+\delta_j}
	 \right]
   \]	
   with
   \begin{equation}\label{eq: delta}
	\delta=(\delta_1,\dots,\delta_k)=(k-1,\dots,0).
   \end{equation}
 \end{itemize}  
\end{lemma}
\begin{proof}
 First, we bring the summation in~\eref{eq: F_k(u,v;b)} out of the determinant,
 \[
  F_k(\mathbf{u},\mathbf{v};\mathbf{b}) =
  \sum\limits_{p_1,\ldots,p_k=0}^{\infty}
   \det\limits_{1\leqslant{i,j}\leqslant{k}}
    \big[b_{p_i}(v_i-u_j)^{p_i}\big].
 \]
 Using the binomial expansion, we obtain
 \begin{eqnarray*}
  F_k(\mathbf{u},\mathbf{v};\mathbf{b}) &=&
  \sum\limits_{p_1,\ldots,p_k=0}^{\infty}\;
   \sum\limits_{q_1=0}^{p_1}\ldots\sum\limits_{q_k=0}^{p_k}
    \det\limits_{1\leqslant{i,j}\leqslant{k}}
     \left[b_{p_i}{p_i\choose{q_i}}v_i^{q_i}(-u_j)^{p_i-q_i}\right] \\ &=&
  \sum\limits_{q_1,\ldots,q_k=0}^{\infty}\;
   \sum\limits_{l_1,\ldots,l_k=0}^{\infty}
    \left(\prod\limits_{i=1}^k b_{q_i+l_i}{q_i+l_i\choose{q_i}}v_i^{q_i}\right)
     \det\limits_{1\leqslant{i,j}\leqslant{k}} \big[(-u_j)^{l_i}\big],
 \end{eqnarray*}
 where in the second line the summation indices are changed to $l_i=p_i-q_i$ and common factors are extracted from the rows of the determinant.
 The independent summations in $q_1,\dots,q_k$ and $l_1,\dots,l_k$ can be subdivided to the summations over ordered tuples and summations over  permutation group $S_k$:
 \begin{eqnarray*}
  F_k(\mathbf{u},\mathbf{v};\mathbf{b}) &=&
  \sum\limits_{0\leqslant q_1\leqslant\ldots\leqslant q_k}^{\infty}\;
  \sum\limits_{0\leqslant l_1<\ldots<l_k}^{\infty}\;
   \sum\limits_{\sigma'\in S_k,\, \sigma'(q)\ne q}\;
    \sum\limits_{\sigma\in S_k} (-1)^{\sigma}
     \det\limits_{1\leqslant{i,j}\leqslant{k}} \big[(-u_j)^{l_i}\big] \\
      &\times& \prod\limits_{i=1}^k b_{\sigma'(q_i)+\sigma(l_i)}
       {\sigma'(q_i)+\sigma(l_i)\choose{\sigma'(q_i)}} v_i^{\sigma'(q_i)},
\end{eqnarray*}
 where the indices $(l_1,\dots,l_k)$ are strictly ordered due to the skew symmetry of the determinant and the permutations of weakly ordered indices $q_1,\dots,q_k$ which leave the tuple unchanged are excluded from the summation.    
Replacing the summation over $\sigma'$ by the summation over shifted permutation $\tau=\sigma^{-1}\sigma'$, we obtain
 \[
  F_k(\mathbf{u},\mathbf{v};\mathbf{b}) =
  \sum\limits_{0\leqslant q_1\leqslant\ldots\leqslant q_k}^{\infty}\;
  \sum\limits_{0\leqslant l_1<\ldots<l_k}^{\infty}\;
   \sum\limits_{\tau\in S_k,\, \tau(q)\ne q}\;
    \sum\limits_{\sigma\in S_k} (-1)^{\sigma}
     \det\limits_{1\leqslant{i,j}\leqslant{k}} \big[(-u_j)^{l_i}\big]
 \]
 \begin{equation}\label{eq: F_k_2}
  {\color{white} F_k(\mathbf{u},\mathbf{v};\mathbf{b})} \times
  \prod\limits_{i=1}^k v_i^{\sigma\tau(q_i)}
   \prod\limits_{i=1}^k b_{\sigma(\tau(q_i)+l_i)}
    {\sigma(\tau(q_i)+l_i)\choose{\sigma\tau(q_i)}}.
 \end{equation}
 Here we should note that, on the one hand,
 \[
  \sum\limits_{\sigma\in S_k} \left(\prod\limits_{i=1}^k v_i^{\sigma\tau(q_i)}\right) (-1)^{\sigma} (-1)^{\tau} = \det\limits_{1\leqslant{i,j}\leqslant{k}} \big[(v_i)^{q_j}\big],
 \]
 and the expression on the r.h.s. is a skew-symmetric function of $q_1,\dots,q_k$. Therefore, it suffices to carry out the summation only for those tuples $q_1,\ldots,q_k$ that satisfy inequalities $0\leqslant q_1<\ldots<q_k$. On the other hand, the second product in~\eref{eq: F_k_2}  does not depend on $\sigma$. That is why, extracting this product, we have
 \[
  \sum\limits_{\tau\in S_k} \left(\prod\limits_{i=1}^k b_{\tau(q_i)+l_i}(r) {\tau(q_i)+l_i\choose{\tau(q_i)}}\right) (-1)^{\tau} = \det\limits_{1\leqslant{i,j}\leqslant{k}} \left[b_{q_j+l_i}(r) {q_j+l_i\choose{q_i}}\right].
 \]
 Consequently, we obtain
 \begin{eqnarray*}
  F_k(\mathbf{u},\mathbf{v};\mathbf{b}) &=&
  \sum\limits_{0\leqslant q_1<\ldots<q_k}^{\infty}
   \sum\limits_{0\leqslant l_1<\ldots<l_k}^{\infty}
    \det\limits_{1\leqslant{i,j}\leqslant{k}}
     \left[b_{q_j+l_i}(r) {q_j+l_i\choose{q_j}}\right]  \\
  &\times& \det\limits_{1\leqslant{i,j}\leqslant{k}} \big[v_i^{q_j}\big]
   \det\limits_{1\leqslant{i,j}\leqslant{k}} \big[(-u_j)^{l_i}\big].
 \end{eqnarray*}
 Finally, we proceed from the summation over strictly ordered $k$-tuples $q=(q_1,\dots,q_k)$ and $l=(l_1,\dots,l_k)$ to the one over weakly ordered partitions $\lambda=q-\delta$ and $\mu=l-\delta$ with $\delta$ defined in~\eref{eq: delta}. Using the definition of Schur symmetric polynomial $s_{\alpha}(\mathbf{x})$ of variables $\mathbf{x}=(x_1,\dots,x_k)$ indexed by a partition $\alpha=(\alpha_1\geqslant\cdots\geqslant\alpha_k\geqslant0)$ in terms of the alternating polynomial~\cite{Macdonald1998},
 \begin{equation}\label{formula: Schur polynomial}
  a_{\alpha+\delta}(\mathbf{x}) =
  \det\limits_{1\leqslant{i,j}\leqslant{k}} \big[x_i^{\alpha_j+\delta_j}\big]
   = \Delta(\mathbf{x}) s_{\alpha}(\mathbf{x}),
 \end{equation}
 we arrive at~\eref{eq: F_k statement}.
\end{proof}

The above lemma can be used to construct asymptotic approximations for the determinants involved in the watermelon probabilities starting from the asymptotic approximations for the Green functions at large distances. Namely, the matrix coefficients are given by Green functions of two arguments associated with positions separated by distances obtained by a finite shift $t$ from a large distance $r$. In other words, we start with an asymptotic expansion of a function of the form $g(r+t)$, which is in fact the expansion in powers of $t/r$. Specifically, the series $f(t)$ represents the asymptotic expansion of some function $g(r+t)$ with respect to the sequence~$\mathbf{b}$  of functions $b_n=b_n(r)$ of a large parameter $r$ in the sense that $b_{n+1}=o(b_n)$, as $r\to \infty$.  Then, the estimate is obtained  using  the following direct consequence of Lemma~\ref{lemma: fundamental determinant lemma}.

\begin{corollary}\label{cor: asym det}
 If the series $f(t)$ given by~\eref{eq: f(t)} represents an asymptotic expansion of a function $g(t+r)$ with respect to an asymptotic sequence $\mathbf{b}$ of functions $b_n=b_n(r)$ of a~variable~$r$ such that $b_{n+1}=O(b_n/r)$, as $r\to \infty$, then 
 \[
  \fl\quad
  \det[g(r+u_i-v_j)]_{1\leqslant i,j\leqslant k} 
  = \Delta(\mathbf{v})  \Delta(-\mathbf{u}) \det_{0\leqslant i,j\leqslant k-1}\left[ b_{i+j}  {i+j \choose{j}} \right] \left(1+O\left(\frac{1}{r}\right)\right),
 \]
as $r\to\infty$, where the equal sign is understood in a sense of asymptotic expansions.
\end{corollary}
\begin{proof}
 Counting powers of $r$ shows that the leading order term of the sum in~\eref{eq: F_k statement}, where 
 \[
  F_k(\mathbf{u},\mathbf{v};\mathbf{b}) =
  \det[g(r+u_i-v_j)]_{1\leqslant i,j\leqslant k},
 \] 
 corresponds to $\lambda=\mu=0^{(k)}$. In this case, $s_{\lambda}(\mathbf{v})=s_{\mu}(-\mathbf{u})=1$, and the power of   $r$ in the correction term is  less by one.  
\end{proof}

Let us consider two basic examples of the use of this statement, relevant for the asymptotics of watermelon probabilities near the open and closed boundary of the half-infinite lattice.

\begin{lemma}\label{lem: determinant for powers}
 If $g(x)=x^{-\alpha}$ with $x\in \mathbb{R}_{\geqslant 0}$ and $\alpha\in \mathbb{R}$, then, as $r\to\infty$,
 \begin{equation}\label{eq: det g power}
  \det_{1\leqslant i,j\leqslant k} [g(r+u_i-v_j)] =
  \frac{\Delta(\mathbf{v}) \Delta(-\mathbf{u})}{r^{k(\alpha+k-1)}}
   \prod_{i=0}^{k-1} \frac{(\alpha)_i}{i!}
    \left(1+O\left(\frac{1}{r}\right)\right),
 \end{equation}
 where $(\alpha)_i = \alpha(\alpha+1)\cdots (\alpha+i-1)$ is the rising factorial aka the Pochhammer symbol.
\end{lemma}
\begin{proof}
 Let us apply Corollary~\ref{cor: asym det} to the function $f(t)$ whose asymptotic expansion is given by
 \begin{equation}\label{eq:g=f}
  f(t) = g(t+r) = \sum_{i=0}^\infty b_n t^n,
 \end{equation}
 where, since $g(x)=x^{-\alpha}$, we have
 \begin{equation}\label{eq: b_n def}
  b_n = \frac{(-1)^n(\alpha)_n}{r^{\alpha+n}n!},\quad
  n\in \mathbb{N}_0. 
 \end{equation}
 The result follows directly from the following chain of identities
 \numparts
 \begin{eqnarray}
  && \!\!\!\! \!\!\!\!
   \det_{0\leqslant i,j\leqslant k-1}
    \left[ b_{i+j}  {i+j \choose{j}} \right] =
   \det_{0\leqslant i,j\leqslant k-1}
    \left[ \frac{(-1)^{i+j}(\alpha)_{i+j}}{r^{\alpha+i+j}(i+j)!}
     {i+j \choose{j}} \right]
       \label{eq: det eq 1} \\
  &&	 = r^{-k(\alpha+k-1)}
   \left(\prod_{i=0}^{k-1} i!\right)^{-2}
    \prod_{i=0}^{k-1} (\alpha)_{i} \,
     \det_{0\leqslant i,j\leqslant k-1} [(\alpha+i)_{j} ]
       \label{eq: det eq 2} \\
  && = r^{-k(\alpha+k-1)}
   \left(\prod_{i=0}^{k-1} i!\right)^{-2}
    \prod_{i=0}^{k-1} (\alpha)_{i} \,
     \det_{0\leqslant i,j\leqslant k-1}
      \left[ \frac{(j-i+1)_i (\alpha+i)_j}{(\alpha+j)_i} \right]
       \label{eq: det eq 3} \\
  && = r^{-k(\alpha+k-1)}
   \prod_{i=0}^{k-1} \frac{(\alpha)_{i}}{i!}
       \label{eq: det eq 4}.  
\end{eqnarray}
 \endnumparts
Here, we first insert definition~\eref{eq: b_n def} of the coefficients $b_n$ into the determinant~\eref{eq: det eq 1}. In~\eref{eq: det eq 2}, using  column- and row-wise linearity of a determinant, we take factors depending  only either on column or row index out of the determinant and apply relation $(\alpha)_{i+j}/(\alpha)_i=(\alpha+i)_j$. In~\eref{eq: det eq 3}, we modify the matrix in the determinant   by replacing the rows below the first one by  their sum with a linear combination of rows above it, 
 \begin{equation}
  \sum _{s=0}^i (-1)^s {i\choose s} (\alpha+i-s)_j =
  \frac{(j-i+1)_i (\alpha+i)_j}{(\alpha+j)_i}
 \end{equation}
which is simplified to a simple ratio of Pochhammer symbols by induction on $i$ or with the help of Chu-Vandermonde identity~\cite{AndrewsAskeyRoy1999}.
Note that $(j-i+1)_i=0$ if $j<i$. Thus, the modified matrix is upper-triangular with diagonal elements with column and row indices $j=0,\dots,k-1$ equal to $j!$. This yields~\eref{eq: det eq 4}.
\end{proof}

Note that the $\alpha=0$ case of the above lemma is trivial, since the only nonzero determinant corresponds to $k=1$. A non-trivial analogue of $\alpha=0$   is given by the log function to which we also  add a constant.

\begin{lemma} \label{lemma: determinant for logarithm}
 If $g(x)=c_1-c_2 \log x$ with $x\in \mathbb{R}_{\geqslant 0}$ and $c_1,c_2 \in \mathbb{R}$, then, as $r\to\infty$,
 \begin{equation}\label{eq: log det asymp}
  \det_{1\leqslant i,j\leqslant k} [g(r+u_i-v_j)] =
  \frac{\Delta(\mathbf{v}) \Delta(-\mathbf{u}) \ln r}{(k-1)! r^{k(k-1)}}
   \left(c_2^{k}+O\left(\frac{1}{\ln r}\right)\right)
 \end{equation}
and 
 \begin{eqnarray}\label{eq: log det asymp const}
  \lim_{c_1\to\infty} c^{-1}_{1}
   \det_{1\leqslant i,j\leqslant k} [g(r+u_i-v_j)] =
  \frac{\Delta(\mathbf{v}) \Delta(-\mathbf{u})}{(k-1)! r^{k(k-1)}}
   \left(c_2^{k-1}+O\left(\frac{1}{r}\right)\right). 
 \end{eqnarray}
The limit in the second identity is supposed to be taken before the limit $r\to\infty$.
\end{lemma}
\begin{proof}
 The statements of this lemma can be obtained from that of Lemma~\ref{lem: determinant for powers} in the limit $\alpha=0$. Specifically, in our case,
 \[
  f(t) = g(t+r) = \sum_{i=0}^\infty \hat{b}_n t^n
 \]
 with 
 \begin{equation}\label{eq: b_n def log}
  \hat{b}_0 = c_1-c_2\ln r, \qquad
  \hat{b}_n = c_2\frac{(-1)^{n} }{nr^{n}}, \quad n\in \mathbb{N}. 
 \end{equation}
 For $n>0$, these coefficients can be obtained from $b_n$ of  \eref{eq: b_n def} as
 \begin{equation*}
  \hat{b}_n = \lim_{\alpha\to 0}\frac{c_2}{\alpha}{b}_n,
 \end{equation*}
 Dividing the r.h.s. of~\eref{eq: det g power} by $\alpha^k$ and considering the limit $\alpha\to 0$, we obtain a finite contribution from parts of the determinant containing only $b_n$ with $n>0$, while the term proportional to $b_0$ is $O(1/\alpha)$. Thus, the result can be obtained by replacing $b_0/\alpha$ by $\hat{b}_0=c_1-c_2\log r$ from \eref{eq: b_n def log}. Note that the term containing $\hat{b}_0$ is, in fact, dominant, though this term depends on the order of limits $r\to \infty$ and $c_1\to \infty$. For finite $c_1$, the term containing $\ln r$ is $O(r^{-k(k-1)}\ln r)$, while the other terms are $O(r^{-k(k-1)})$. This yields~\eref{eq: log det asymp}. On the other hand, if we divide the result by $c_1$ and take the limit $c_1\to \infty$ first, only the terms proportional to $c_1$ survive, resulting in~\eref{eq: log det asymp const}. Finally, note that since the only singular in $\alpha$ terms  are those containing $b_0$, the limit $\alpha \to 0$ preserves the order of corrections.
\end{proof}

\section{Watermelon probabilities}\label{sec: Isotropic Watermelons}

In this section, we conclude the proof of Theorem~\ref{th: watermelon probabilities asymp} and obtain the asymptotic probabilities of a watermelon configuration near the open and closed boundaries.

\subsection{Open boundary}
Following~\eref{eq: watermelon probability in general}, we express the watermelon probability as a ratio of two determinants:
\begin{equation}
\label{formula: watermelon through det/det in isotropic case op}
 \mathbb{P}_{\mathcal{L}^{\mathrm{op}}}^{I\cup\partial^{\mathrm{op}}}(\mathfrak{F}(IJ|\partial^{\mathrm{op}})) =
 \frac{\det_{1\leqslant i,j\leqslant k}\big[G^{\mathrm{op}}_{(r+v_i-u_j;1,1)}\big]}{\det_{1\leqslant i,j\leqslant k}\big[G^{\mathrm{op}}_{(i-j;1,1)}\big]}.
\end{equation}
Below, we estimate the asymptotics of the numerator at large $r$ and evaluate explicitly the coefficients of the matrix under the determinant in the denominator. The asymptotic behavior of the numerator follows from Lemma~\ref{lem: determinant for powers} with $\alpha=2$ and $v_i=k+1-i$, $u_i=i$, where $i=1,\dots,k$. More precisely, asymptotics~\eref{eq: approximation for G^(op)_(r;1,1)} for $G^{\mathrm{op}}_{(r+v_i-u_j;1,1)}$ suggests
 \[
  g(r) = \frac{1}{\pi r^2}\left(1+O\left(\frac{1}{r}\right)\right).
 \]
Taking into account that
 \[
  \Delta(\mathbf{v}) = \Delta(-\mathbf{u}) = \prod_{i=0}^{k-1}i!,
 \]
Lemma~\ref{lem: determinant for powers} yields
 \begin{equation}
  \det_{1\leqslant i,j\leqslant k}\big[G^{\mathrm{op}}_{(r+v_i-u_j;1,1)}\big] =
  \frac{1}{\pi^{k} r^{k(k+1)}}
   \prod_{i=0}^{k-1} i!(i+1)!
    \left(1 + O\left(\frac{1}{r}\right)\right).
 \end{equation}
On the other hand, the denominator $\det_{1\leqslant i,j\leqslant k}\big[G^{\mathrm{op}}_{(i-j;1,1)}\big]$ is expressed in terms of the Green functions evaluated at finite distances separating sites within the string~$I$. The  double integral~\eref{formula: integral for G^(cl)_(x,y)} representing $G^{\mathrm{op}}_{(n;1,1)}$ for $n\in\mathbb{Z}$ can be evaluated  to a finite sum 
 \begin{eqnarray*}
  G^{\mathrm{op}}_{(n;1,1)} &=&
  G^{\mathrm{op}}_{(-n;1,1)} \\ &=&
  2\delta _{n,0} - \frac{1}{2}\delta_{|n|,1} -
   \sum_{s=0}^{|n|} \sum_{r=0}^s (-1)^{r+s}
    {2|n|\choose 2s} {s\choose r} f_{-}\big(|n|-s+r\big).
 \end{eqnarray*}
Here, the quantities $f_-(m)$  (as well as $f_+(m)$ appearing below in the case of the closed boundary) are defined in terms of specific values of the hypergeometric functions
 \begin{equation}\label{eq: f_pm}
  f_{\pm}(m) =
  \frac{8^{\pm 1/2}}{\pi }
   \frac{_2 F_1\left(\pm\frac{1}{2},m+\frac{1}{2},m+\frac{3}{2};\frac{1}{2}\right)}{\left(m+\frac{1}{2}\right)}.
 \end{equation}
For non-negative integer values of $m$, the values of $f_-(m)$ can further be represented as another finite sum including gamma-functions of integer and half-integer arguments only:
 \[ 
 \begin{array}{ccll}  
  f_{-}(m) &=&
  \displaystyle\frac{1}{2\sqrt{\pi}}
   \sum_{l=0}^{m-1} (-1)^{l} {m-1\choose l}
    \frac{ \Gamma\left(\frac{l}{2}+\frac{1}{2}\right) }
    		{ \Gamma\left(\frac{l}{2}+2\right)}
    		 &\quad m\geqslant 1 \\
   &=& \displaystyle\left(1+\frac{2}{\pi}\right)
   &\quad m=0.
 \end{array}
 \]
As a result, the values of the Green function of interest are given by a sum of a rational number and a rational multiple of $1/\pi$, e.g.
 \[
  G^{\mathrm{op}}_{(n;1,1)} \; = \;
  1 - \frac{2}{\pi}, \;\;
  \frac{2}{\pi} - \frac{1}{2}, \;\;
  \frac{10}{3\pi} - 1, \;\;
  \frac{38}{3\pi} - 4, \;\;
  \frac{802}{15\pi} - 17, \;\;
  \frac{1194}{5\pi} - 76, \;
  \dots
 \]
for $n=0,1,2,3,4,5$ etc. These values are to be substituted into the determinant in the denominator of~\eref{formula: watermelon through det/det in isotropic case op}. As a result, the values of the determinants are given by  polynomials in $\pi^{-1}$ with rational coefficients. The examples for small values of $k$ are shown in Table~\ref{tab: det op}. 
\begin{table}[ht]
\caption{\label{tab: det op} The values of the determinant in the denominator of~\eref{formula: watermelon through det/det in isotropic case op} of the probability coefficient of $k$-leg watermelon near the open boundary for $k=1,\dots,5$.}
\begin{indented}
\item[]
\renewcommand{\arraystretch}{2}
	\begin{tabular}{ccc}
		\br 
		$k$ & \quad & $\det_{1\leqslant i,j \leqslant k}\big[G^{\mathrm{op}}_{(i-j;1,1)}\big]$ \\  \mr
		1 & \quad & $\displaystyle 1-\frac{2}{\pi}$ \\  
		2 & \quad & $\displaystyle \frac{3}{4}-\frac{2}{\pi }$ \\  
		3 & \quad & $\displaystyle -1+\frac{40}{3 \pi }-\frac{448}{9 \pi ^2}+\frac{512}{9 \pi ^3}$\\  
		4 & \quad & $\displaystyle -\frac{435}{16}+\frac{1843}{6 \pi }-\frac{11584}{9 \pi ^2}+\frac{64000}{27 \pi ^3}-\frac{131072}{81 \pi ^4}$ \\  
		5 & \quad & $\displaystyle -\frac{8075}{16}+\frac{155293}{24 \pi }-\frac{7333616}{225 \pi ^2}+\frac{401408}{5 \pi ^3}-\frac{194510848}{2025 \pi ^4}+\frac{268435456}{6075 \pi ^5} $ \\  \br
	\end{tabular}
\end{indented}
\end{table}

Finally, we conclude that the probability of the $k$-leg watermelon near the open boundary is given by~\eref{eq: open} with
 \begin{equation}\label{eq: C_k^op}
  C_k^{\mathrm{op}} =
  \frac{\prod_{i=0}^{k-1}i!(i+1)!}
   {\pi^k\cdot\det_{1\leqslant i,j \leqslant k}
    \big[G^{\mathrm{op}}_{(i-j;1,1)}\big]}.
 \end{equation}
This concludes the proof of the first statement of Theorem~\ref{th: watermelon probabilities asymp}.

\subsection{Closed boundary}
For the closed boundary conditions, formula~\eref{eq: watermelon probability in general} also suggests that the watermelon probability is a ratio of two determinants:
 \begin{equation}\label{formula: watermelon through det/det in isotropic case cl}
  \mathbb{P}_{\mathcal{L}^{\mathrm{cl}}}^{I}(\mathfrak{F}(IJ)) =
  \frac{\det_{1\leqslant i,j\leqslant k}\big[G^{\mathrm{cl}}_{(r+v_i-u_j;1,1)}\big]}
   {\det_{1\leqslant i,j\leqslant k}\big[G^{\mathrm{cl}}_{(i-j;1,1)}\big]}.
 \end{equation}
However, now the entries of matrices contain the infinite constant. In this case, both the numerator and denominator are infinitely large and their leading terms are proportional to the infinite constant, which cancels within the ratio that has a finite limit. Thus, it is enough to find these leading terms in both numerator and denominator. 

To estimate the numerator, we apply Lemma~\ref{lemma: determinant for logarithm}. From asymptotics~\eref{eq: approximation for G^(cl)_(r;1,1)} for $G^{\mathrm{cl}}_{(r+v_i-u_j;1,1)}$, we have
 \[
  g(r) = 2G_{0,0} - \frac{\ln{r}}{\pi} + O(1).
 \]
Hence, it follows from Lemma~\ref{lemma: determinant for logarithm} that
 \[
  \det_{1\leqslant i,j\leqslant k}
   \big[G^{\mathrm{cl}}_{(r+v_i-u_j;1,1)}\big] \simeq
  \frac{2G_{0,0}}{\pi^{k-1} r^{k(k-1)}}
   \prod_{i=1}^{k-1} i!(i-1)!
    \left(1+O\left(\frac{1}{r}\right)\right),
 \]
where the sign ``$\simeq$'' indicates that we keep only the diverging part proportional to $G_{0,0}$ and neglect the finite part of the Green function.

To evaluate the determinant $\det_{1\leqslant i,j\leqslant k}\big[G^{\mathrm{op}}_{(i-j;1,1)}\big]$ in the denominator, we recall that the Green function is a sum of infinite and finite parts,
 \[
  G^{\mathrm{cl}}_{(n;1,1)} = 2G_{0,0}+g_{\mathrm{fin}}^{\mathrm{cl}}(|n|),
 \]
where the finite part has the double integral representation~\eref{formula: integral for G^(cl)_(x,y)} and can be reduced to a sum  
 \begin{eqnarray*}
  \fl g_{\mathrm{fin}}^{\mathrm{cl}}(n) =
  &-&\frac{1}{2} \delta _{n,0} + \left(f_+(n)-\sum _{s=0}^{n-1} f_+(s)\right) \\ \fl
  &-&\sum_{s=1}^n \left(\sum _{r=0}^{s-1} (-1)^{r+s-1}
   {2n\choose 2s} {s-1\choose r} (2 f_+(n+r-s)-f_+(n+r-s+1)\right). 
 \end{eqnarray*}	
The summands in the above formula are defined in terms of the quantity $f_+(m)$ from~\eref{eq: f_pm} that can be represented as another finite sum for non-negative integer values of $m$:
 \begin{eqnarray}
  f(m) &=& 
  \frac{1}{4\sqrt{\pi}} \sum_{l=0}^m
   \frac{(-1)^l \Gamma \left(\frac{l}{2}+\frac{1}{2}\right)}
    {\Gamma \left(\frac{l}{2}+1\right)} {k\choose l},
     \quad m\geqslant 0. 
\end{eqnarray}
Like those of $G^{\mathrm{op}}_{(n;1,1)}$, the values of $g_{\mathrm{fin}}^{\mathrm{cl}}(n)$ are linear in $1/\pi$ with rational coefficients:
 \[
  g_{\mathrm{fin}}^{\mathrm{cl}}(n) \; = \;
  -\frac{1}{4}, \;\;
  -\frac{1}{4} - \frac{1}{\pi}, \;\;
  -\frac{3}{4}, \;\;
  \frac{13}{3\pi} - \frac{9}{4}, \;\;
  \frac{64}{3\pi} - \frac{31}{4}, \;\;
  \frac{459}{5\pi} - \frac{121}{4}, \;\;
  \dots 
 \]
for $n=0,1,2,3,4,5$ etc. The determinant in the denominator of~\eref{formula: watermelon through det/det in isotropic case cl} can be transformed to the determinant of a block matrix, in which all the dependence on the infinite $G_{0,0}$ part has been moved to a single element in the upper left corner:
 \begin{eqnarray*}
  \det_{1 \leqslant i,j \leqslant k}
   \big[G^{cl}_{(i-j;1,1)}\big] &=&
  \det_{1 \leqslant i,j \leqslant k}
   \big[2G_{0,0}+g_{\mathrm{fin}}^{\mathrm{cl}}(i-j)\big] \\ &=&
  \det\left(
   \begin{tabular}{cc}
    $2G_{0,0}-g_{\mathrm{fin}}^{\mathrm{cl}}(0)$ & $*$ \\
    $*$ & $B_{k-1}$
   \end{tabular}
  \right) \simeq
  2G_{0,0} \det B_{k-1}.
 \end{eqnarray*}
The part of this determinant proportional to the infinite part is given in terms of its principle minor, aka the determinant of $(k-1)\times(k-1)$ symmetric square matrix $B_{k-1}$, where the entries of $B_s$ for any positive integer $s$ are given by
 \[
  [B_{s}]_{ij} =
  g_{\mathrm{fin}}^{\mathrm{cl}}(|i-j|) - 
  g_{\mathrm{fin}}^{\mathrm{cl}}(i) - 
  g_{\mathrm{fin}}^{\mathrm{cl}}(j) + 
  g_{\mathrm{fin}}^{\mathrm{cl}}(0),
  \qquad 1 \leqslant i,j \leqslant s.
 \]
Similarly to the open boundary case, the values of the determinants are given by polynomials in $\pi^{-1}$ with rational coefficients. The examples for small values of $k$ are shown in Table~\ref{tab: det cl}. 
\begin{table}[ht]
	\caption{\label{tab: det cl} The values of the determinant in the denominator of~\eref{formula: watermelon through det/det in isotropic case cl} of the probability coefficient of $k$-leg watermelon near the closed boundary for $k=1,\dots,5$.}
\begin{indented}
\item[]
\renewcommand{\arraystretch}{2}
	\begin{tabular}{ccc}
		\br 
		$k$ & \quad & {$\det B_{k-1}$} \\  \mr
		1 & \quad & $1$ \\  
		2 & \quad & $\displaystyle \frac{2}{\pi }$ \\  
		3 & \quad & $\displaystyle -\frac{1}{4}+\frac{2}{\pi }$ \\  
		4 & \quad & $\displaystyle \frac{1}{2}-\frac{26}{3 \pi }+\frac{128}{3 \pi ^2}-\frac{512}{9 \pi ^3}$ \\  
		5 & \quad & $\displaystyle -\frac{7}{16}+\frac{145}{6 \pi }-\frac{896}{3 \pi ^2}+\frac{33280}{27 \pi ^3}-\frac{131072}{81 \pi ^4}$ \\  \br
	\end{tabular}\renewcommand{\arraystretch}{1}
\end{indented}
\end{table}

Finally, we conclude that the probability of the $k$-leg watermelon near the closed boundary is given by~\eref{eq: closed} with
 \begin{equation}\label{eq: C_k^cl}
  C_k^{\mathrm{cl}} =
  \frac{ \prod_{i=1}^{k-1}i!(i-1)! } {\pi^{k-1}\cdot \det B_{k-1}}.
 \end{equation}
This proves the second statement of Theorem~\eref{th: watermelon probabilities asymp}.

\begin{remark}\label{rem}
 As we noted in the beginning, the arguments similar to those that we used for the watermelons near the closed boundary are also applicable to the watermelons in the bulk considered in~\cite{IvashkevichHu2005, GorskyNechaevPoghosyanPriezzhev2013}. Specifically, in the bulk, the watermelon probability defined by~\eref{eq: watermelonprob} would also be given by the ratio of two determinants. Although this probability has a more complicated structure beyond the realm of applicability of our Lemma~\ref{lemma: fundamental determinant lemma}, it is still a sum of infinite and finite parts. Similarly to the ones described by the statements of Lemma~\ref{lemma: determinant for logarithm}, the leading behavior of the infinite part of the numerator has a~power law distance dependence, while the finite one has the logarithmic prefactor. As a~consequence, only the power law 
part survives after the normalization by the likewise infinite denominator. At the same time, the quantity calculated in  \cite{IvashkevichHu2005,GorskyNechaevPoghosyanPriezzhev2013} is a finite part of the denominator that has the form of the logatrithm times power law. This is the source of the discrepancy between the two sets of results discussed in the introduction. 
\end{remark}

\section*{Acknowledgments}
 The problem studied here was proposed to us by Vyatcheslav Priezzhev and Philippe Ruelle. We thank them for stimulating discussion on the subject. We also thank Sergei Nechaev for useful discussions. KhN thanks Andrea Sportiello for discussion and indicating a shorter proof of a particular case of Lemma~\ref{lem: determinant for powers}. AP thanks Eveliina Peltola for discussion and providing useful references. AP thanks Guillaume Barraquand for attracting his attention to ref.~\cite{Fomin2001}. The article is supported by the Russian Foundation for Basic Research under grant 20-51-12005.

\section*{References}

\bibliographystyle{ieeetr}
\bibliography{iopart-num}

\end{document}